\documentclass[letterpaper,twocolumn,10pt]{article}
\usepackage{usenix-2020-09}
\usepackage{bbm}
\usepackage{xurl}
\usepackage{amsthm}
\usepackage{amsmath}
\usepackage{amsfonts}
\usepackage[shortlabels]{enumitem}
\usepackage{stmaryrd}
\usepackage{xspace}
\usepackage{multirow}
\usepackage{graphicx}
\usepackage{mathtools}
\usepackage{standalone}
\usepackage{inconsolata}
\usepackage{algorithmic}
\usepackage{mdframed}
\newmdenv[skipabove=4mm, skipbelow=4mm, leftmargin=4mm, rightmargin=4mm]{mybox}
\newtheorem{theorem}{Theorem}[section]
\usepackage[linesnumbered]{algorithm2e}
\usepackage{subcaption}
\usepackage{authblk}
\setenumerate[1]{itemsep=0pt,partopsep=0pt,parsep=\parskip,topsep=1pt}
\setitemize[1]{itemsep=0pt,partopsep=0pt,parsep=\parskip,topsep=1pt}
\setdescription[1]{itemsep=0pt,partopsep=0pt,parsep=\parskip,topsep=1pt}

\newcommand\name{Spin\xspace}

\newcommand\ztwol{\mathbb{Z}_{2^l}}
\newcommand{\para}[1]{\smallskip\noindent\textbf{#1}}

\newcommand{\code}[1]{\texttt{#1}\xspace}

\newcommand{\share}[2]{\llbracket #1 \rrbracket ^\text{#2}}
\RestyleAlgo{ruled}
\SetKw{KwBy}{by}
\SetKwInput{KwPara}{Require}
\DeclarePairedDelimiter\floor{\lfloor}{\rfloor}
\DeclarePairedDelimiter\ceil{\lceil}{\rceil}

\usepackage{tikz}

\date{}

\title{\name: An Efficient Secure Computation Framework with GPU Acceleration}
\author[1]{Wuxuan Jiang\thanks{jiangwuxuan@gmail.com Contribute equally to this paper.}}
\author[1]{Xiangjun Song\thanks{eatenbagpipe@gmail.com Contribute equally to this paper.}}
\author[1]{Shenbai Hong\thanks{hongshenbai@tsingj.com}}
\author[1]{Haijun Zhang\thanks{zhanghaijun@tsingj.com}}
\author[1]{\authorcr Wenxin Liu\thanks{liuwenxin@tsingj.com}}
\author[1]{Bo Zhao\thanks{zhaobo@tsingj.com}}
\author[2]{Wei Xu\thanks{weixu@tsinghua.edu.cn}}
\author[1]{Yi Li\thanks{xiaolixiaoyi@gmail.com}}
\affil[1]{Tsingjiao Information Technology Co. Ltd.}
\affil[2]{Tsinghua University}

\date{\today}

\begin{document}
	
	\maketitle
	
	\begin{abstract}

  Accuracy and efficiency remain challenges for multi-party computation (MPC) frameworks.
  \name is a GPU-accelerated MPC framework that supports multiple computation parties and a dishonest majority adversarial setup. 
  We propose optimized protocols for non-linear functions that are critical for machine learning, as well as several novel optimizations specific to attention that is the fundamental unit of Transformer models,
  allowing \name to perform non-trivial CNNs training and Transformer inference without sacrificing security.
  At the backend level, \name leverages GPU, CPU, and RDMA-enabled smart network cards for acceleration.
  Comprehensive evaluations demonstrate that \name can be up to $2\times$ faster than the state-of-the-art for deep neural network training.
  For inference on a Transformer model with 18.9 million parameters, our attention-specific optimizations enable \name to achieve better efficiency, less communication, and better accuracy.
  
\end{abstract}

	\section{Introduction}
As artificial intelligence keeps shaping our daily lives~\cite{rombach2022high,openai2023gpt4,wang2023neural}, data security and privacy concerns become more important than ever, especially in scenarios where data comes from multiple parties.
A typical example is joint training of deep neural networks, which requires lots of data from different sources.
Another example is the secure inference of pre-trained models, which needs to protect both the data and the models.
Several promising technologies can contribute to resolving the dilemma of machine learning and security.  
Multi-party computation (MPC)~\cite{yao1982protocols,goldreich1987play,ben1988completeness} enables multiple parties to compute collaboratively using their inputs without compromising privacy.
Nonetheless, MPC incurs substantial computation and communication costs.  
Federated learning~\cite{mcmahan2017communication}, on the other hand, distributes a training task across multiple devices and adjusts the joint model collectively using local gradients.
Unfortunately, this optimization diminishes the level of privacy protection, leading to problems such as data poisoning and inference attacks~\cite{tolpegin2020data,nasr2019comprehensive}.
Differential privacy~\cite{dwork2014algorithmic,abadi2016deep} is another widely used technology that regulates the privacy level by adjusting the amount of noise, which may reduce its precision for high privacy levels.
In this paper, we focus on enhancing the efficiency of MPC-based approaches with provable security properties.


Non-linear functions widely appear in neural networks.
However, there exist several challenges in optimizing non-linear functions for MPC:
1) As most MPC frameworks use fixed-point numbers to represent real numbers, 
evaluating a non-linear function using na\"ive iterative methods might lead to inaccurate results.
2) Transformer models, such as \cite{graham2021levit,openai2023gpt4,devlin2018bert}, heavily rely on the attention mechanism which is constructed using approximations for non-linear functions like \code{Exponentiation}~\cite{li2022mpcformer,dong2023puma},
and there have not been solutions for performing joint optimizations across operations within an attention block to achieve optimal performance.
3) Some frameworks~\cite{knott2021crypten, watson2022piranha} improve the efficiency by exposing part of intermediate results to plaintext (e.g., Piranha~\cite{watson2022piranha} evaluates the reciprocal in plaintext when computing \code{Softmax}), which poses potential security risks.
Finding the optimal balance between security and efficiency poses a challenging task in this particular context, given the inherent difficulty in quantifying the impact of exposure on overall end-to-end security.

One natural strategy to accelerate MPC is to utilize GPUs' enormous parallel threads to handle large inputs.
However, secret sharing protocols necessitate many rounds of communication among parties, and transferring data between GPU memory and CPU memory incurs significant latency,
pulling down the acceleration benefit of GPUs~\cite{keller2022secure}.
Furthermore, some existing GPU-accelerated frameworks such as~\cite{tan2021cryptgpu, wagh2021falcon} use ABY3-style protocols~\cite{mohassel2018aby3}, which is based on \emph{honest-majority} (no two parties can collude).
People require a general multi-party computation that supports an arbitrary number of parties with a \emph{dishonest majority}.
Piranha~\cite{watson2022piranha}, a cutting-edge framework, though claiming that it supports the general case theoretically, has not been implemented yet.

We design and implement \name, an MPC framework with a dishonest majority that enables efficient training and inference of deep neural networks.
Like Piranha~\cite{watson2022piranha}, we concentrate on GPU-acceleration-capable protocols.
\name allows for faster and more accurate computations for larger deep learning models.
\name achieves these desirable properties through a collection of new designs at both the protocol and implementation levels,
We further optimize the secure inference of Transformer models by capturing common patterns of attention blocks.

In this study, we introduce a novel collection of non-linear function algorithms, including \code{Reciprocal}, \code{Exponentiation}, and \code{Logarithm}.
These algorithms enhance the efficiency of neural network training in \name.
We also use several novel attention-specific optimizations to achieve more accurate and faster Transformer inference.

At the implementation level, we employ GPUs for acceleration and design a double-ring buffer to enable asynchronous GPU-CPU memory copy and reduce the latency in the computation's critical path.
We also support \emph{remote direct memory access} (RDMA)~\cite{rdma_wiki} to accelerate communication on supported hosts (otherwise, \name reverts to the socket to remain compatible).
For computation-light but communication-intensive operations, we devise a CPU-GPU hybrid computation model that avoids unnecessary GPU-CPU memory copy and only offloads computation-intensive operations to GPUs.
For multi-head attention that employs independent matrix multiplications, we divide GPU cores into different groups and assign matrix multiplications to core groups with a carefully tuned group size.

We have implemented \name on a real-world testbed and trained four deep neural networks on real-world datasets. 
We show that, even delivering stronger security, \name may still be up to $2\times$ faster and up to 15\% more accurate than the current state-of-the-art on the same hardware.

For inference, we demonstrate that \name outperforms neural networks of different sizes, including a Transformer model with 18.9 million parameters (at 3.08 samples per second).


In summary, our contributions are as follows:
\begin{enumerate}[leftmargin=0.4cm]
\item We propose a series of algorithms for non-linear functions, achieving better accuracy and efficiency than existing ones.
\item We propose a combination of implementation-level optimizations that effectively utilize modern GPU and RDMA hardware to accelerate MPC computation.
      As a by-product, we add RDMA support to Piranha~\cite{watson2022piranha}.
\item We design a set of novel optimizations specific to attention, reducing the number of several key operations like secure multiplication in attention blocks.
\item Using \name, we implement several non-trivial deep learning models and demonstrate significant performance improvements in efficiency, communication, and accuracy.
\end{enumerate}

	\section{Related Work}
\label{sec:related}

SecureML~\cite{mohassel2017secureml} is among the pioneers linking MPC to machine learning.
It uses local shift for truncation and employs a combination \code{ReLU}s to approximate \code{Sigmoid} and \code{Softmax},
which works well for both training and inference of simple machine learning models such as logistic regression and simple neural networks.
Following SecureML, \cite{rathee2020cryptflow2} provides more efficient protocols for comparison and division, while \cite{rathee2021sirnn} gives a comprehensive library of non-linear functions for machine learning.
Both of them focus on addressing secure inference.
However, SecureML-style protocols have not been proven effective for training complex neural networks.

MPC frameworks for machine learning use different security settings.   
SecureML's line of work only supports two parties.
Another line of work, pioneered by ABY3~\cite{mohassel2018aby3}, relies on an honest-majority 3-party protocol
that adopts mixed circuit computation~\cite{demmler2015aby} and offers efficient protocols for basic operations like multiplication.
Recent works like \cite{wagh2021falcon,tan2021cryptgpu,watson2022piranha,keller2022secure} are built upon the same 3-party setting.
There are also 4-party frameworks such as Trident~\cite{chaudhari2019trident}, FLASH~\cite{byali2019flash}, FantasticFour~\cite{dalskov2021fantastic} and PrivPy~\cite{li2019privpy}.
Crypten~\cite{knott2021crypten} is one of the few frameworks that support general multi-party secure machine learning, i.e. supporting $n$-party dishonest majority, as \name does.
MP-SPDZ~\cite{keller2020mp} provides a variety of MPC protocols with varying levels of security, including the \code{Semi2k} protocol that supports dishonest majority.
A few works exclusively focus on 2-party secure inference. LLAMA~\cite{gupta2022llama} proposes a framework based on Function Secret Sharing (FSS) and relies on a trusted dealer.
Cheetah~\cite{huang2022cheetah} combines several MPC protocols to support efficient neural network inference with well-designed linear operations and non-linear operations like \code{ReLU} and truncation,
but it lacks the support for neural network training and non-linear functions like \code{Softmax} commonly used in Transformer models. 


There have been various prior attempts to use GPUs to accelerate MPC.
CryptGPU~\cite{tan2021cryptgpu} is a three-party honest majority framework that uses GPUs to accelerate neural network training. It is based on Crypten~\cite{knott2021crypten}.
However, CryptGPU still lacks the accuracy to train a deep neural network from scratch, as both \cite{watson2022piranha} and \cite{keller2022secure} have pointed out. 
Piranha~\cite{watson2022piranha} is a GPU-assisted secure machine learning framework that supports two, three, and four parties by implementing SecureML~\cite{mohassel2017secureml},
Falcon~\cite{wagh2021falcon}, and FantasticFour~\cite{dalskov2021fantastic}.
Piranha provides cutting-edge performance, demonstrating for the first time that training a VGG16-style deep neural network can be completed in about a day. Piranha can also be deployed on multiple GPUs\cite{mengmulti}.
However, Piranha does not provide accurate enough non-linear functions, lowering overall training quality, and reveals intermediate results in \code{Softmax} in exchange for performance.

Non-linear function evaluation plays a key role in neural network training and inference.
Secure non-linear function evaluation usually requires polynomial approximation, and frameworks adopt different approaches.
For example, Falcon~\cite{wagh2021falcon} proposes protocols for various non-linear functions. Some of the protocols, however, exhibit privacy leaks.
SECFLOAT~\cite{rathee2022secfloat} switches to floating-point computation for accurate and fast evaluation.
Although they produce outstanding results for numerous non-linear functions,  the high cost of addition limits their practical application.
Crypten~\cite{knott2021crypten} introduces simple approximation techniques with limited accuracy.
NFGen~\cite{fan2022nfgen} proposes an auto-fitting framework for a piece-wise polynomial approximation solution.
However, it necessitates a limited range of each function and batched comparison operations to choose the appropriate coefficients for each evaluation.
Moreover, NFGen does not support multi-variant functions like \code{Softmax}.

Recently, some works have managed to run secure inference on Transformer models.
One of the bottlenecks they encounter is the \code{Softmax} evaluation.
\cite{keller2020effectiveness} compares the effectiveness of different \code{Softmax} replacement strategies.
MPCFormer~\cite{li2022mpcformer} uses an aggressive \code{Softmax} replacement strategy called \emph{2Quad approximation}.
With the help of knowledge distillation, MPCFormer can overcome the accuracy loss caused by the rough approximation.
As the other side of the coin, the rough approximation of MPCFormer necessitates retraining models to fit the aggressive alteration and can hardly enable general inference using pre-trained models.
PUMA~\cite{dong2023puma} modifies the \code{Softmax}-\code{ReLU} replacement strategy and uses piecewise polynomials to approximate \code{GeLU}.
PUMA can evaluate an LLaMA-7B inference pass in 5 minutes and is the state-of-the-art of secure inference of large language models to our knowledge.
Our evaluations will show that our optimized attention outperforms the one based on PUMA's \code{Softmax} in both efficiency and accuracy.
	\section{System Overview}\label{sec:system}


\subsection{Threat Model}
We assume semi-honest adversaries, where each party follows the protocol but is curious about the other parties' sensitive information.
\name works in a \emph{dishonest majority} setting, i.e., for a secure computation task with $n$ parties, \name can tolerate up to $n-1$ parties' collusion.

\subsection{System Design}

Fig.~\ref{fig:architecture} provides an overview of \name,
and illustrates the workflow of a secure computation task.
\name consists of several modules, and the following five modules are key to \name.  

\begin{figure}[tb]
\centering
\includegraphics[width=0.5\textwidth]{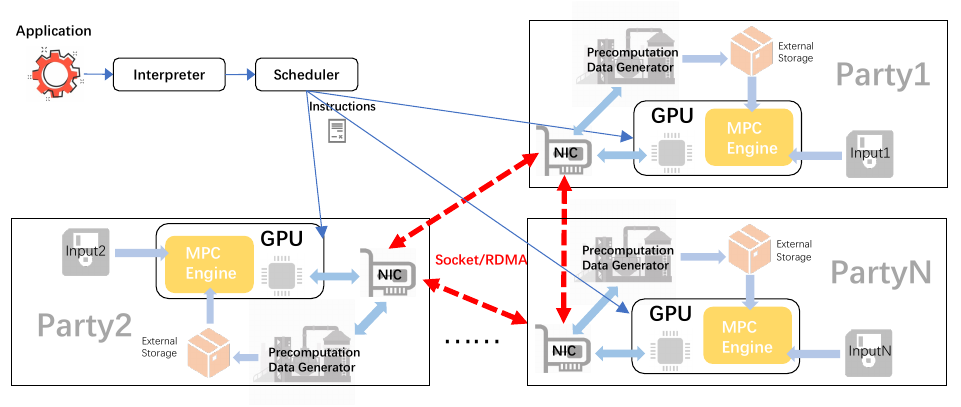}
\vspace{-0.10in}
\caption{The architecture of \name. There can exist multiple parties with individual inputs in each task.
         }
\vspace{-0.10in}
\label{fig:architecture}
\end{figure}

\para{Module 1: Application Programming Interface (API).}
Similar to many extant frameworks, \name enables programmers to write code in a high-level language akin to Python.  
The APIs are composed of three layers:
a) basic operations, including addition, multiplication, and comparison;
b) derived operations, such as non-linear functions and other operations that build on basic operations;
c) library algorithms, such as neural network training.
A programmer can either use basic operations and derived operations to implement customized algorithms for specific applications or use the higher level \name library in suitable situations.

\para{Module 2: Interpreter and \name Instructions.}
\name offers a suite of lower-level APIs, called \emph{\name Instructions}, which precisely depict the flow of a \name program using low-level engine-friendly operators.
The interpreter reads the user-level program and generates a sequence of \name instructions.
The design of the interpreter is beyond the scope of this paper.
\name instructions comprise three types:

\begin{itemize}[leftmargin=0.4cm]
    \item \code{Input/Output}: Input instructions drive the engine to read input data from data sources and secretly share the inputs with the other parties, while output instructions tell the engine to reveal results to specified receivers.
    \item \code{Memory manipulation}: The engine locally transforms the secret shares it holds.
        	Operations include \emph{matrix reshaping}, \emph{memory copy}, \emph{array concatenation}, and more.
    \item \code{Secure computation}: The engine uses MPC protocols to carry out computation on secretly shared data.
    The engine automatically establishes communication when evaluating operations that require data transfer,
    such as multiplication, exponentiation, etc.
	Applications do not need to concern themselves with when to send or receive data.
\end{itemize}

\para{Module 3: Scheduler. }
Upon receiving the instructions for a task, the scheduler generates the task description, which includes the instructions and participant information.
The scheduler then inserts the task description into a task backlog. When computational resources become available, the scheduler assigns the task to each party.

\para{Module 4: MPC Engine. }
The engine is responsible for executing MPC protocols, storing secret shares in GPU memory, receiving instructions from the scheduler, and executing the instructions using GPU cores. 
Each engine receives the same set of task descriptions and executes them synchronously and iteratively.
For each instruction, the engines invoke the corresponding MPC protocol and communicate with one another over network channels.

\para{Module 5: Precomputation Data Generator (PDG).}
To accelerate the online computation, we implement an input-independent offline phase to precompute a portion of data such as beaver triples~\cite{beaver1992efficient}.
We have a per-party PDG to produce precomputation data via oblivious transfer~\cite{naor2001efficient} and store it in disk-based storage.


\para{Hardware Accelerators. }
In \name, we use GPUs for computation acceleration and SmartNICs for efficient communication.
To achieve high-bandwidth and low-latency networking, we further employ \emph{Remote Direct Memory Access} (RDMA)~\cite{rdma_wiki} to offload data movement among machines from CPU cores.
We implement two variants of the engine: one on CPUs and another on GPUs.
The latter offloads some lightweight but latency-sensitive computation to CPUs for acceleration.

	\section{Building Blocks}
\label{sec:block}

We discuss the fundamental components of our MPC operations in this section.
We first introduce basic concepts, then list all the basic operations used in \name.

\subsection{Notations}
The set $\ztwol$ represents the integers in $[-2^{l-1}, 2^{l-1}-1]$.
The arithmetic computation is in $\ztwol$ by default in this paper.

For a signed integer $x \in \ztwol$, we use $\share{x}{A}$ to represent the $n$-out-of-$n$ arithmetic sharing of $x$.
Specifically, in a $n$-party setting, $\share{x}{A}$ means $x$ is split into $n$ shares $x_1, x_2, \dots, x_n \in \ztwol$, such that $x = \sum_{i=1}^n x_i$.
Similarly, for a boolean number $x \in \{0, 1\}$, the binary sharing $\share{x}{B}$ means $x$ is split to $n$ shares $x_1, x_2, \dots, x_n \in \{0, 1\}$
such that $x = \bigoplus_{i=1}^n x_i$, where $\bigoplus$ stands for \emph{exclusive OR} (XOR).
$\share{x_{0:l-1}}{B}$ is a list of $l$ binary sharings of bits of $x$, from $\share{x_0}{B}$ to $\share{x_{l-1}}{B}$.

We encode real numbers to fixed-point numbers of bit length $l$. A real number $x$ is encoded into a signed integer $\floor{2^px} \in \ztwol$ with precision $p$.
We use $[x]_p$ to represent the encoded integer of $x$. For example, $[7.7]_{p}=7.7\times 2^p$.
We use $\share{x,p}{A}$ for the sharing of real numbers, i.e., $\share{x,p}{A}$ = $\share{[x]_p}{A}$.
The operations on $\share{x,p}{A}$ work on $\share{[x]_p}{A}$ by default.


%

\subsection{Precomputation Data}\label{sec:precompute}
The engine can speed up online computation using input-independent precomputation data. There are three main types of precomputation data in \name:
\emph{beaver triples}~\cite{beaver1992efficient}, \emph{doubly-authenticated bits} (\emph{daBits})~\cite{rotaru2019marbled} and \emph{extended doubly-authenticated bits} (\emph{edaBits})~\cite{escudero2020improved}.
Previous approaches~\cite{escudero2020improved,beaver2001precomputing,keller2018overdrive,keller2016mascot} use either oblivious transfer (OT) or homomorphic encryption to generate these types of precomputation data.
\name uses OT to generate them. 
The process of generating precomputation data has a different pattern compared with online computation.
In this paper, we focus on optimizing the online phase of complex neural networks and leave the optimization for precomputation as future work.

\subsection{Basic Operations}
\label{sec:basic}

We enumerate the basic operations on which our nonlinear functions algorithms are built.
These operations are implemented following the \code{Semi2k} protocol of MP-SPDZ~\cite{keller2020mp}.

\para{Addition.}
$\share{z}{A} \leftarrow \code{Add}(\share{x}{A}, \share{y}{A})$, where $z = x + y$.
The arithmetic sharing approach allows additive homomorphism in nature, and we can conduct addition locally.
Similarly, in a binary sharing scheme, the parties can perform \code{XOR} locally.

\para{Scaling.}
$\share{z}{A} \leftarrow \code{Scale}(k, \share{x}{A})$, where $k$ is a public number and $z = k \times x$.

\para{Multiplication.}
$\share{z}{A} \leftarrow \code{Mul}(\share{x}{A}, \share{y}{A})$, where $z = x \times y$.
The multiplication between two shared numbers $\share{x}{A}$ and $\share{y}{A}$ demands a beaver triple to accomplish the online phase.
For the multiplication of fixed-point numbers, \code{Mul} doubles the precision,
i.e., $\share{c, 2p}{A} \leftarrow \code{Mul}(\share{a, p}{A}, \share{b, p}{A})$ where $[c]_{2p} = [a]_p \times [b]_p$.
Then we perform truncation on $\share{c, 2p}{A}$ to get $\share{c, p}{A}$ (we will introduce our truncation protocol in Section~\ref{sec:derived}).
We similarly evaluate matrix multiplication and convolution, using corresponding beaver triples.

\para{Bit to arithmetic.}
$\share{z}{A} \leftarrow \code{Bit2A}(\share{x}{B})$, where $z = x$.
This operation converts a secret bit to a secret integer and requires a precomputed daBit.

\para{Binary to arithmetic.}
$\share{z}{A} \leftarrow \code{Compose}(\share{x_{0:l-1}}{B})$, where $z=\sum_{i=0}^{l-1}2^ix_i$.
We construct this operation using $l$ independent \code{Bit2A} operations.

\para{Arithmetic to binary.}
$\share{z_{0:l-1}}{B} \leftarrow \code{Decompose}(\share{x}{A})$, where $x=\sum_{i=0}^{l-1}2^iz_i$.
We use edaBits to implement it.

\para{Comparison.}
$\share{z}{B} \leftarrow \code{Gt}(\share{x}{A}, \share{y}{A})$, where $x > y$ indicates $z = 1$, otherwise $z = 0$.
We evaluate this operation by first computing $\code{Decompose}(\share{x}{A}, \share{y}{A})$, then retrieving the most significant bit as the sign of $(x-y)$.
It is worth noting that this protocol may be incorrect if $(x-y)$ overflows in the ring,
and we can avoid such risk using two extra \code{Decompose} operations.
In most machine learning tasks, however, we can skip the two extra \code{Decompose}, since the number ranges of machine learning can hardly cause such overflow.

\subsection{Derived Operations}
\label{sec:derived}
We can use the above basic operations to construct the following derived operations.
Here we only introduce the operations that can be constructed na\"ively and leave the introduction of non-linear functions in the next section.


\para{Truncation.}
$\share{z}{A} \leftarrow \code{Truncate}(\share{x}{A}, f)$, where $z = x\cdot 2^{-f}$.
Multiplication of shared fixed-point numbers requires truncation to correct the precision.
A common way is to invoke $\code{Decompose}(x)$ and shift the bits, then invoke $\code{Compose}$ to get the result.
However, $\code{Decompose}$ involves binary addition, which requires expensive multi-round communication.
In a 2-party setting, each party can perform bit shift locally, without any communication.
But local shift causes overflow with probability proportional to the absolute value of the shared number~\cite{mohassel2017secureml}.
For example, if we truncate the result after one multiplication on $\mathbb{Z}_{2^{64}}$ with precision $p=26$, the failure probability will be as high as $2^{-11}$ (or $0.5\text{\textperthousand}$).
\cite{mohassel2018aby3} also points out that the local shift cannot work in a 3-party setting.

The authors of \cite{dalskov2020secure} give a solid probabilistic truncation protocol for unsigned integers, which we denote as \code{UnsignedTruncate}.
The output integer of \code{UnsignedTruncate} probabilistically deviates from the accurate result by $1$.
This precision loss is negligible, specifically for machine learning tasks,
as deviation by $1$ in fixed-point representation only causes a precision loss of $2^{-p}$.
We extend \code{UnsignedTruncate} to a signed version for general computation.
Speficially, for a signed integer $x$, we first apply $\share{y}{A} \leftarrow \code{UnsignedTruncate}(\share{x}{A} + 2^{l-2}, l, f)$,
then set the result as $\share{y}{A} - 2^{l-f-2}$.
The correctness is direct:
for $x \in [-2^{l-2}, 2^{l-2}-1]$, $x + 2^{l-2}$ is non-negative and can be treated as an unsigned number,
then we have $y = (x + 2^{l-2}) / 2^f$, thus $y - 2^{l-f-2}$ is the result $x / 2^f$.
For $x < -2^{l-2}$ or $x > 2^{l-2}-1$, $x + 2^{l-2}$ may be negative, thus cannot be treated as unsigned, making $\code{UnsignedTruncate}$ not work.
In machine learning tasks, we can safely assume that $x$ always falls in the range $[-2^{l-2}, 2^{l-2}-1]$, making sure that the output of \code{Truncate} is correct.
Our evaluations in Section~\ref{sec:experiment} demonstrate that \code{Truncate} with such assumption achieves desirable results.

%
%
%
%

\para{Polynomial evaluation.}
$\share{z}{A} \leftarrow \code{Evaluate}(\share{f(x)}{A})$, with $z = f(x)$ and the coefficients of $f(x)$ are public.
As a minor optimization, we can save truncations when computing $a_ix^i$. 
For instance, $f(x)=a_3x^3+a_2x^2+a_1x+a_0$ is evaluated by sequentially computing $x,x^2,x^3$.
We can skip truncations when computing $a_ix^i$, and sum up each term directly.
Then a single truncation suffices to get the result.

\para{Left-most one (LMO) extraction.}
$\share{z_{0:l-1}}{B} \leftarrow \code{LMO}(\share{x_{0:l-1}}{B})$, where there is only a bit of $1$ in $z_0, z_1, \dots, z_{l-1}$ and $\sum_{i=0}^{l-1}2^i z_i \leq \sum_{i=0}^{l-1}2^ix_i < \sum_{i=0}^{l-1}2^{i+1}z_i$.
A non-negative value's LMO is the first $1$ after leading 0s.
Our LMO extraction algorithm comes from \cite{catrina2010improved}:
First, we scan the bit sequence and run a \code{pre-OR} protocol to turn all bits after LMO into 1,
then we shift one bit and compute \code{XOR} with the original bit sequence,
after which all bits after LMO become 0 while LMO stands out.

%
%
	\section{Non-linear Functions Evaluation}
\label{sec:non_linear}

Non-linear functions often face efficiency and accuracy issues in secure neural network training and inference.
This section presents our new algorithms for non-linear functions.
It is important to note that these algorithms are based on the operations detailed in Section~\ref{sec:block} and are independent of their implementations, which implies:
a) our algorithms remain secure as long as the operations in Section~\ref{sec:block} are universally composable;
b) optimizations or new implementations for these operations maintain the correctness of our new algorithms.
In this paper, we focus on three non-linear functions: \emph{reciprocal},  \emph{exponentiation}, and \emph{logarithm}, which are frequently used in neural networks.



Researchers have proposed various approaches to calculating or approximating non-linear functions.
For instance, the Newton-Rhapson method is a common approach to approximate reciprocal.
Such approximation methods are widely adopted by recent frameworks, including~\cite{li2019privpy,tan2021cryptgpu,knott2021crypten,mohassel2017secureml,mohassel2018aby3}.
The use of these approaches in complex machine learning tasks, however, presents two obstacles: 
a) some functions such as the exponential function $e^x$ output values in vast ranges, causing polynomials to fit poorly;
b) piecewise functions, while effective in addressing accuracy issues, need expensive range detection to select the appropriate polynomial.


In this paper, we base our methods on the ones from~\cite{catrina2010secure,keller2022secure,aly2019benchmarking}.
These methods all share the concept of mapping the input value into a relatively narrow range before applying numerical techniques.
This input mapping strictly constrains the input domain for the subsequent numerical approach, resulting in stable numerical behavior and improved accuracy.
We further optimize these methods by reducing costly operations and using superior numerical methods.

\subsection{Reciprocal}
A commonly used approach for secure reciprocal evaluation is the Newton-Raphson method~\cite{enwiki:newton}.
However, this method requires a good initial guess for convergence, which limits its application.
Crypten~\cite{knott2021crypten} uses a heuristic function to generate the initial guess in a specific input domain.
The heuristic function, however, does not have a demonstrable error bound, and its evaluation is costly.
Furthermore, Crypten's solution cannot handle negative values, and an additional comparison is needed.
Falcon~\cite{wagh2021falcon} estimates the rough range of the input by revealing the position of LMO, leaking extra information.

Algorithm~\ref{alg:ssnr} shows our algorithm for reciprocal.
For the Newton-Raphson iteration $x_{n+1} = x_n(2 - zx_n)$, we set initial guess $x_0=1$. According to the convergence criterion, $0<x_0<\frac{2}{z}$ is required while $z\in[0.5,1)$.
Let $q=1-z$, then we have $x_1 = 1+q$ and $x_2=x_1(2 - zx_1)=(1+q)(1+q^2)$. Through induction, we can derive
$h(z)=x_n=\prod_{i=0}^{n-1} \Big(1+q^{2^i}\Big)=\prod_{i=0}^{n-1} \Big(1+(1-z)^{2^i}\Big)$.
The number of iteration rounds $d$ can be set flexibly according to the tradeoff between accuracy and efficiency.
Our idea is that we first obtain the absolute value of input $x$ (line~\ref{rec:absolute}) and scale $|x|$ to the range $[0.5,1)$ by finding the LMO (line~\ref{rec:lmo} - \ref{rec:truncate1}),
then evaluate the numerical method on the scaled value and finally scale the result back (line~\ref{rec:evaluate} - \ref{rec:truncate2}).
Note that, despite restricting the scaling factor in $[2^{-p}, 2^{p-1}]$ (line~\ref{rec:swap_start} - \ref{rec:zero}),
meaning we can only scale the input in range $[2^{-p}, 2^p)$ to $[0.5,1)$,
we argue that the correctness is unaffected, as the reciprocal of a number greater than $2^p$ underflows the fixed-point representation with precision $p$.

Our idea is similar to \cite{catrina2010secure}, with several optimizations.
First, the sign bit extraction of $x$ shares the secure bit decomposition with LMO (line~\ref{rec:decompose}, \ref{rec:bit2a} and \ref{rec:lmo}),
making the sign bit extraction free,
while \cite{catrina2010secure} invokes an extra secure comparison between $x$ and $0$.
Second, despite that, the standard way for computing the negative value of a signed integer is to flip the bits and add $1$,
we omit the addition and only perform the bit flip when $x$ is negative (line~\ref{rec:xor_start} - \ref{rec:xor_end}), thus saving a binary addition.
This omission only introduces an extra error of $2^{-p}$ for fixed-point numbers, which is negligible in most machine learning tasks.

\begin{algorithm}[tb]
\small
\caption{Reciprocal.}\label{alg:ssnr}
\KwData{$\share{x,p}{A},x\neq 0$,\text{ bit length }$l$,\text{ precision }$p$}
\KwPara{Newton-Raphson iteration function $h(z)$}
\KwResult{$\share{y,p}{A}\approx \share{1/x,p}{A}$}
$\share{x_{0:l-1}}{B} \leftarrow \code{Decompose}(\share{x,p}{A})$\label{rec:decompose}\\
\For{$i \gets 0$ \KwTo $l-2$} { \label{rec:xor_start}
    $\share{x_i}{B} \leftarrow \code{XOR}(\share{x_{l-1}}{B}, \share{x_i}{B})$\label{rec:xor}\\
} \label{rec:xor_end}
$\share{s}{A} \leftarrow \code{Bit2A}(\share{x_{l-1}}{B})$\label{rec:bit2a}\\
$\share{x_{+},p}{A}\leftarrow \share{x,p}{A}-2\share{s}{A}\cdot\share{x,p}{A}$\label{rec:absolute}\\
$\share{x_{0:l-2}}{B}\leftarrow\code{LMO}(\share{x_{0:l-2}}{B})$\label{rec:lmo}\\
\For{$i \gets 0$ \KwTo $p-1$} {\label{rec:swap_start}
    $\code{Swap}(\share{x_i}{B},\share{x_{2p-1-i}}{B})$\\
}\label{rec:swap_end}
$\share{x_{2p}}{B}, \share{x_{2p+1}}{B}, \dots, \share{x_{l-1}}{B} \leftarrow 0$ \label{rec:zero} \\
$\share{t,p}{A} \leftarrow \code{Compose}(\share{x_{0:2p-1}}{B})$\label{rec:compose}\\
$\share{y,p}{A} \leftarrow \code{Truncate}(\share{t,p}{A}\cdot\share{x_{+},p}{A},p)$\label{rec:truncate1}\\
$\share{y,p}{A} \leftarrow \code{Evaluate}(\share{h(y),p}{A})$\label{rec:evaluate}\\
$\share{y,p}{A} \leftarrow \code{Truncate}(\share{t,p}{A}\cdot\share{y,p}{A},p)$\label{rec:truncate2}\\
$\share{y,p}{A} \leftarrow \share{y,p}{A}-2\share{s}{A}\cdot\share{y,p}{A}$\label{rec:flip_back}\\
\end{algorithm}

\subsection{Exponentiation}
Many activation functions in machine learning, such as \code{Sigmoid} and \code{Softmax}, significantly rely on exponentiation.
Approximation approaches techniques based on Taylor series~\cite{taylor1717methodus} or Remez approximation~\cite{remez1934determination} work well for limited input ranges
but have accuracy issues for inputs with a wide range of variations.
Taylor's theorem states that the approximation is only adequate near a specific point, so \cite{keller2022secure} applies the approximation to the fractional portion of the input rather than relying entirely on the approximation as~\cite{tan2021cryptgpu} does.
For example, in deep neural network training, the inputs of \code{Softmax} may range from $-30$ to $30$, causing the output of exponentiation to vary from $e^{-30}$ to $e^{30}$.

The algorithm for exponentiation proposed by \cite{keller2022secure} enhances the one in \cite{aly2019benchmarking}, and
we will call this algorithm as \emph{SOTA-exp}.
Due to the page limit, we present the pseudocode of SOTA-exp in Appendix~\ref{appendix:baseline_exp}.
The basic idea is that we can convert the evaluation of $e^x$ to $2^z$ where $z = x\cdot\log_2{e}$.
The benefit is that, for a positive number $z$,
we can evaluate $2^{z_{int}}$ by invoking \code{Compose} where $z_{int}$ is the integer part of $z$, and approximate $2^{z_{dec}}$ with a polynomial for the decimal part of $z$.
In this case, we only need to approximate the decimal portion in $[0, 1]$ with polynomials.
The algorithm starts with changing the base to 2 by multiplying $\log_2 e$ and a truncation.
 

Algorithm~\ref{alg:ssne} shows our algorithm.  It is based on SOTA-exp but with four major differences.
The first is about how to correct the result of negative inputs.
In line~\ref{exp:addp} of Algorithm~\ref{alg:ssne}, we add a $p$ to $x'$ so that values in $[-p,0)$ become non-negative, we then use truncation to scale back the result (line~\ref{exp:result}).
If $x' < -p$, we set the final result to 0 (line~\ref{exp:sign} - \ref{exp:result}), since the result will be less than $2^{-p}$ and underflows the fixed-point representation with precision $p$.
In comparison, SOTA-exp performs a secure comparison using binary addition,
then converts the result to arithmetic sharing for multiplication.
Thus, our strategy saves a binary addition and a \code{Bit2A}.

Second, we keep the precision of the product of line~\ref{exp:product} as $2p$, and do the truncation together with the bit decomposition.
With this trick, this truncation becomes free.

Third, as SOTA-exp uses a truncation of $2^{c}$ bits to correct results of negative inputs,
it sets $l=31, p=16$ to avoid underflow caused by truncation, and thus it can only process inputs under $(l-p)\ln2 = 15\ln2$.
Our algorithm, on the contrary, exploits the full bit length $l=64$ and handles larger inputs.

Finally, instead of the Talyor expansion used by SOTA-exp, we use the Remez algorithm to obtain a more precise approximation.
In particular, the Remez algorithm has consistent error rates within the specified range compared to Taylor expansion.
We list the coefficients obtained by the Remez algorithm in Appendix~\ref{appendix:coefficients}.


\begin{algorithm}[tb]
\small
\caption{Exponentiation.}\label{alg:ssne}
\KwData{$\share{x,p}{A}$,\text{ bit length }$l$,\text{ precision }$p$}
\KwPara{Remez approximation $f_1(z)$}
\KwResult{$\share{y,p}{A}\approx \share{e^x,p}{A}$}
$\share{x',2p}{A}\leftarrow\share{x,p}{A}\cdot[\log_2 e]_p$\label{exp:product}\\
$\share{x'',2p}{A}\leftarrow [p]_{2p}+\share{x',2p}{A}$\label{exp:addp}\\
$\share{x_{0:l-1}}{B}\leftarrow\code{Decompose \& Truncate}(\share{x'',2p}{A},p)$\label{exp:decompose}\\
$c\leftarrow\ceil{\log_2(l-p-1)}$ \label{exp:intbit}\\
$\share{x_{p:p+c-1}}{A}\leftarrow\code{Bit2A}(\share{x_{p:p+c-1}}{B})$\\
$\share{I}{A}\leftarrow\prod_{i=0}^{c-1} \Big( 2^{2^{i}}\share{x_{p+i}}{A}-\share{x_{p+i}}{A}+1\Big)$\\
$\share{t,p}{A}\leftarrow\code{Compose}(\share{x_{0:p-1}}{B})$\\
$\share{y,p}{A}\leftarrow\share{I}{A}\cdot\code{Evaluate}(\share{f_1(t),p}{A})$ \label{exp:evaluate}\\
$\share{s}{A}\leftarrow 1-\code{Bit2A}(\share{x_{l-1}}{B})$\label{exp:sign}\\
$\share{y,p}{A}\leftarrow\code{Truncate}(\share{s}{A}\cdot\share{y,p}{A},p)$\label{exp:result}
\end{algorithm}

\subsection{Logarithm}
We follow the classic fast algorithm from Goldberg~\cite{fastlog} and convert it into a secure version.
The procedure begins by scaling the input value to the interval $[0.75,1.5)$ before calculating $\log_2 x$. After that, we calculate $\ln 2 \log_2 x$ as the result $\ln x$.
Specifically, for $x \in [0.75, 1.5)$, we calculate logarithm using a 4th-degree polynomial $f_2(z)$. We list the coefficients of $f_2(z)$ from \cite{fastlog} in Appendix~\ref{appendix:coefficients}.

Algorithm~\ref{alg:ssnl} illustrates the full algorithmic.
To scale $x$ to $[0.75, 1.5)$, we first map $x$ to $[0.5,1)$ (line~\ref{log:decompose}, \ref{log:lmo}, and line~\ref{log:swap_start} - \ref{log:swap_end}),
using the scaling technique in Algorithm~\ref{alg:ssnr},
then verify whether $x<0.75$. If that is true, we double $x$.
Checking the following bit of LMO (line~\ref{log:offset_start} - \ref{log:offset_end}) yields the bit $r$ that indicates whether $x<0.75$ or not.
A quick example for $x=50$: its LMO represents $2^5$ since $2^5<x<2^6$.
The bit next to LMO determines whether $x>2^5+2^4$ or not. So the scaled $x$ becomes $50/(2\times 2^5)>0.75$. 
Since the scaling factor is a power of 2, we can easily calculate its logarithm. We conclude by combining a polynomial approximation for scaled input with the logarithm of the scaling factor.

\begin{algorithm}[tb]
\small
\caption{Logarithm.}\label{alg:ssnl}
\KwData{$\share{x,p}{A}$,\text{ bit length }$l$,\text{ precision }$p$}
\KwPara{Remez approximation $f_2(z)$}
\KwResult{$\share{y,p}{A}\approx \share{\ln x,p}{A}$}
$\share{x_{0:l-1}}{B}\leftarrow\code{Decompose}(\share{x,p}{A})$ \label{log:decompose}\\
$\share{x'_{0:2p-1}}{B}\leftarrow\code{LMO}(\share{x_{0:2p-1}}{B})$ \label{log:lmo}\\
$\share{x''_{0:2p-2}}{B}\leftarrow\share{x'_{1:2p-1}}{B},\share{x''_{2p-1}}{B}\leftarrow 0$\label{log:offset_start}\\
$\share{r}{B}\leftarrow\bigoplus_{i=0}^{2p-1}(\share{x_{i}}{B}\land\share{x''_{i}}{B})$\label{log:offset_end}\\
$\share{x_{0:2p-1}}{A}\leftarrow\code{Bit2A}(\share{x_{0:2p-1}}{B})$\\
$\share{r}{A}\leftarrow\code{Bit2A}(\share{r}{B})$\\
$\share{y_L}{A}\leftarrow\sum_{i=0}^{2p-1}(i-l)\share{x_{i}}{A}$\\
\For{$i \gets 0$ \KwTo $p-1$} { \label{log:swap_start}
    $\code{Swap}(\share{x_i'}{B},\share{x_{2p-1-i}'}{B})$\\
} \label{log:swap_end}
$\share{m,p}{A}\leftarrow\code{Compose}(\share{x_{0:2p-1}'}{B})$\\
$\share{x,p}{A}\leftarrow 2\share{x,p}{A}-\share{x,p}{A}\cdot\share{r}{A}$\\
$\share{t,p}{A}\leftarrow \code{Truncate}(\share{x,p}{A}\cdot\share{m,p}{A},p)-2^p$\\
$\share{y_R,p}{A}\leftarrow\code{Evaluate}(\share{f_2(t),p}{A})$\\
$\share{y,p}{A}\leftarrow(\share{y_L,p}{A}+\share{y_R,p}{A}+\share{r}{A}\cdot 2^p)[\ln 2]_p$\\
$\share{y,p}{A}\leftarrow \code{Truncate}(\share{y,p}{A},p)$\\
\end{algorithm}

The above algorithm does not handle inputs exceeding $2^{p}$. We address this issue with an extra comparison.
When receiving an input, we truncate it by $l-2p$ bits and apply Algorithm~\ref{alg:ssnl}, then add $l-2p$ back to the result.
The comparison is used to distinguish the two cases.
We allow users to decide whether to perform this comparison in practice.

Prior work~\cite{aly2019benchmarking} proposes an algorithm based on scaling the input into $[0.5,1)$ and applying a division of polynomials to calculate the logarithm for values in $[0.5,1)$.
This division not only substantially increases computational costs but also introduces inaccuracy because the division relies on reciprocal, which are challenging to compute precisely.
Errors accumulate through two polynomial evaluations and the division. 
In contrast, our double scaling technique confines the input to the range [0.75, 1.5), enabling the subsequent polynomial approximation to be both precise and efficient in practice.

\subsection{Security Analysis}

The above algorithms use the operations in Section~\ref{sec:basic} and \ref{sec:derived} as building blocks,
and those operations work with proven security.
Therefore, we can formalize our algorithms using the \emph{arithmetic black box model} ($F_{ABB}$-$Hybrid$ model)~\cite{damgaard2003universally},
and prove the security of our algorithms under the framework of \emph{Universal Composition}~\cite{canetti2001universally}.
Appendix~\ref{appendix:security} shows the details of the formal proof.


	\section{Hardware Acceleration and Optimizations}
\label{sec:gpu}

Secure computation is typically much more intensive than its plaintext counterpart.
Secure comparison, for instance, requires expensive binary additions and multi-round communication.
Deep neural networks also typically employ large-scale matrix multiplication and convolution, requiring lots of secure 64-bit or 128-bit matrix multiplications and truncations.
\name uses GPUs for accelerating computation, and offloads networking overhead to SmartNICs.

\subsection{Computation Acceleration using GPUs}
GPUs are frequently used to accelerate highly parallelized computations, particularly matrix multiplication, and convolution.
A GPU contains thousands of CUDA cores for general computation and hundreds of tensor cores for mixed-precision computation. 
As secret-sharing-based protocols usually use 64-bit or 128-bit integers for fixed-point representation,
CryptGPU~\cite{tan2021cryptgpu} uses an integral component of four floating-point numbers to construct a 64-bit integer
and employs tensor cores to accelerate operations on floating-point numbers.
Piranha~\cite{watson2022piranha}, on the contrary, implements integer kernels using CUTLASS~\cite{CUTLASS},
and their evaluations demonstrate that their integer kernels perform better than those used in CryptGPU.
We also use CUTLASS for matrix multiplication and convolution but implement the GEMM kernels using the built-in \texttt{int64\underline{{ }}t} and \texttt{\underline{{ }{ }}int128} type of CUDA11.6~\cite{cuda11}.

In detail, we implement the integer kernel for matrix multiplication with the following considerations:
\begin{itemize}[leftmargin=0.4cm]
    \item In GPUs, the threads within a thread block can access the same shared memory, whereas all threads from distinct blocks can only share a global memory.  Shared memory is significantly faster than global memory, but much smaller.  
          Thus we can use shared memory as the cache of global memory to reduce the overall data access latency by carefully planning when and what data we need to move from global memory data into shared memory.
    \item As access to the shared memory of GPUs is significantly faster than access to the global memory, we break up the matrix into blocks with a carefully chosen size that just fits into shared memory and performs block multiplications in parallel.
          In this way, each thread block has its copy of the blocks, avoiding global memory access and contention.
    \item We construct a pipeline for data loading and computation to exploit the computational power of GPUs.
          Specifically, we carefully adjust the size of each block so that the latency of loading block matrices from global memory into shared memory matches the time required for the block matrix multiplication, keeping the pipeline full most of the time. 
\end{itemize}

\subsection{Communication using SmartNICs}
\label{sec:comm}


Communication occupies heavily in MPC. High-throughput and low-latency communication channels are urgent for large-scale secure computation.
Some SmartNICs support \emph{remote direct memory access} (RDMA)~\cite{rdma_wiki}, enabling data copy from the memory of one computer to the memory of another, bypassing the operating systems and CPUs.
We use RDMA to accelerate communication and introduce several optimizations to overcome the challenges of using RDMA in practice.

\begin{figure}[tb]
\centering
\includegraphics[width=0.35\textwidth]{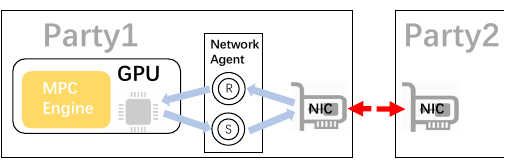}
\vspace{-0.10in}
\caption{Pipelined data transfer.
         The ring labeled `S' represents the sending data ring buffer, while the ring labeled `R' represents the receiving data ring buffer.}
\vspace{-0.10in}
\label{fig:buffer}
\end{figure}

\para{Pipelining data transfer to reduce overall latency.}
We can implement data transfer to a remote GPU in a na\"ive manner:
First, the source host gathers all the data to be sent from the GPU and stores it in the host's memory. Next, the source host transmits the data over the network channel. Finally, the remote host performs the reverse operation.
The end-to-end latency consists of three parts: two data movements between GPU memory and host memory via PCIe, and a data transfer through the network.
If the network channel is a socket channel over TCP/IP, the data movement latency via PCIe is negligible because it is typically much less than the socket latency.
In the case of RDMA, however, the data movement latency cannot be overlooked.
Given that the bandwidth of RDMA can reach 200 Gbps in our testbed, the throughput of the end-to-end data transfer from one GPU to another may be roughly $3\times$ slower than RDMA's bandwidth.

To tackle this issue, we introduce a per-host \emph{network agent} to asynchronously manage PCIe data movement and network data transfer, as illustrated in Figure~\ref{fig:buffer}.
When ready to transmit data, the engine divides it into fixed-size blocks and inserts them into the sending data buffer block by block.
A network agent worker thread constantly checks and transmits new data, while another worker thread stores newly arrived blocks and waits for the engine to retrieve them.
The entire data transfer process is thus pipelined, making the end-to-end throughput close to the RDMA bandwidth.

\para{Allocating ring buffers in pinned memory. }
To minimize the cost of frequent memory allocation and deallocation,
a network agent features two \emph{ring buffers}: one for sending data and another for receiving data.
We implement the ring buffers using \emph{CUDA pinned memory}
to bypass the operation system's virtual memory system so that data can be moved directly between GPU memory and host memory with the help of DMA.
The price of pinned memory is more initialization time and physical memory space,
as pinned memory uses locked pages and fixed physical memory addresses.
However, ring buffers use fixed memory regions in nature, without the concern of frequently allocating new memory space,
thus can fully enjoy the advantage of pinned memory.

\para{Using sockets over TCP/IP for small data transfer.}
As RDMA bypasses CPUs, the main program does not know when a transfer job finishes and usually needs an assistant channel for notification (see \cite{gdrcopy, rdma-core} for examples).
Therefore, RDMA cannot consistently outperform sockets over TCP/IP with an excess number of small data transfers.
For data size smaller than a given threshold, \name chooses a socket channel over TCP/IP, instead of an RDMA channel, for data transfer.
In our environment, we set the threshold to 0.125MB.


\subsection{CPU-GPU Hybrid Computation}
\label{sec:hybrid}

As mentioned earlier, the high latency associated with data copying between GPU memory and host memory may offset the benefits of GPU acceleration.
For computation-intensive tasks, such as the secure multiplication of large matrices, GPU acceleration is advantageous due to the prohibitive computation time.
However, for tasks requiring numerous communication rounds, each involving minimal computation and only a few small data blocks for transfer, using GPUs can reduce overall performance since the cost of data copy dominates.
For example, the \code{Softmax} function falls into this category.

We tackle this issue by employing a \emph{CPU-GPU hybrid mode}.
\name utilizes a few CPU cores (e.g., 20 cores) to handle operations characterized by lightweight computation, frequent communication, and small network packets, while GPUs manage the rest.
This approach achieves improved performance without incurring additional costs.
Furthermore, since using RDMA in practice still requires a socket for synchronization, employing RDMA for communication becomes superfluous for frequent, small package transfers.
Consequently, we utilize sockets for communication incurred by CPU activities and RDMA for GPU operations.

	\section{Attention-Specific Optimizations}
\label{sec:attention}

\begin{figure}[tb]
  \begin{subfigure}[t]{.23\textwidth}
    \centering
    \includegraphics[width=\linewidth]{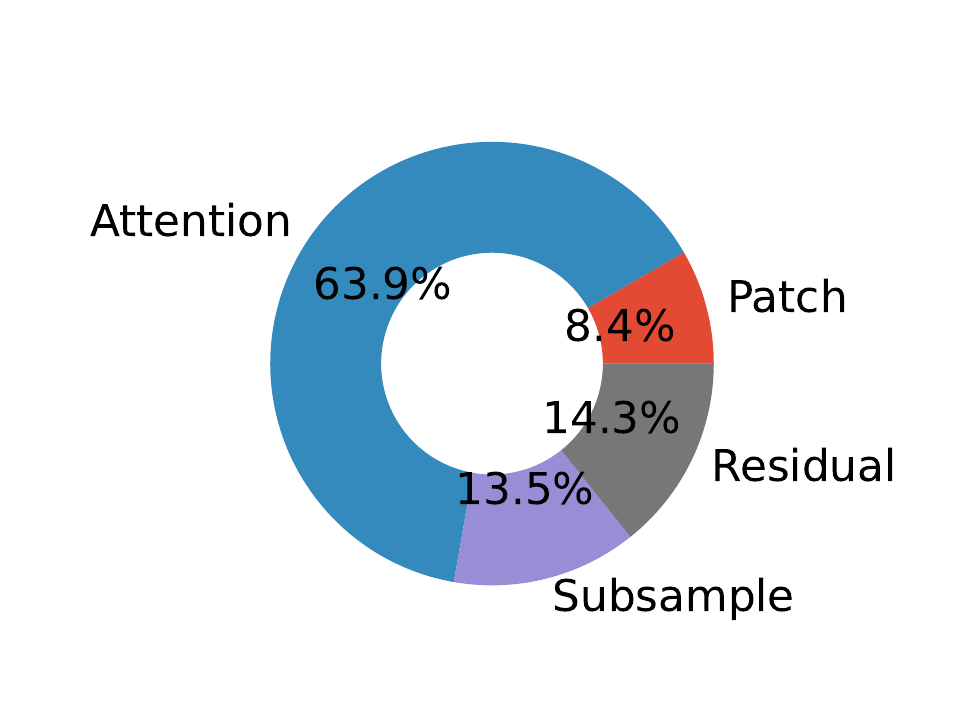}
    \vspace{-0.40in}
    \caption{LeViT-256 inference}
    \vspace{-0.10in}
    \label{fig:levit_breakdown}
  \end{subfigure}
  \hfill
  \begin{subfigure}[t]{.23\textwidth}
    \centering
    \includegraphics[width=\linewidth]{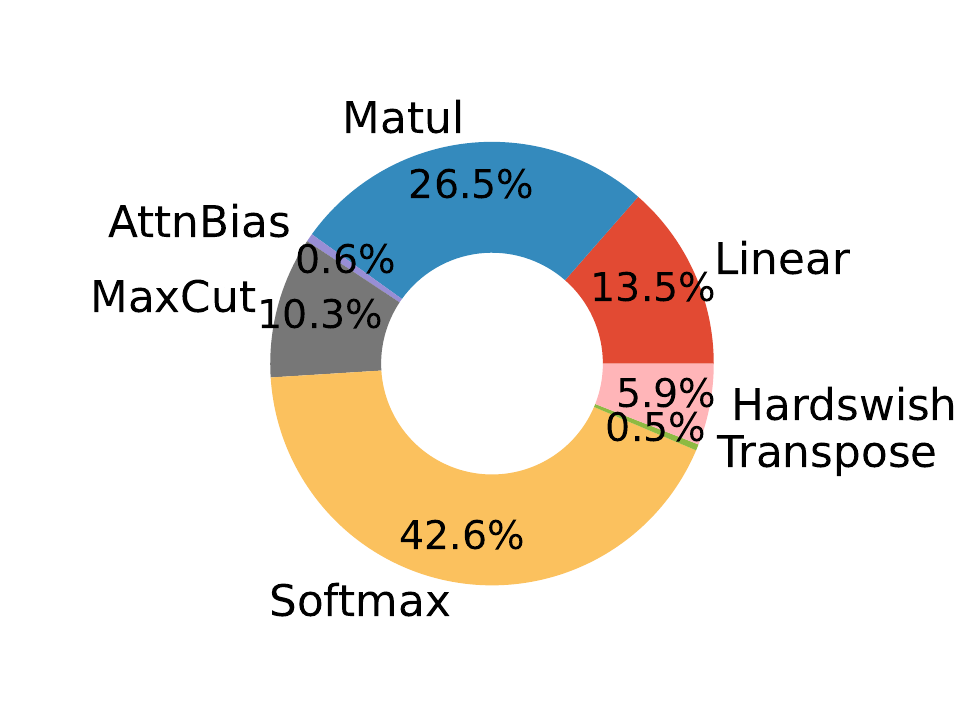}
    \vspace{-0.40in}
    \caption{Attention block}
    \vspace{-0.10in}
    \label{fig:attention_breakdown}
  \end{subfigure}
  \vspace{-0.10in}
  \caption{Runtime breakdown.}
  \vspace{-0.10in}
  \label{fig:breakdown}
\end{figure}

The attention mechanism is the fundamental building block of Transformer models, such as LeViT~\cite{graham2021levit}, GPTs~\cite{radford2019language,brown2020language} and BERT~\cite{devlin2018bert}.
An attention block involves a batch \code{Softmax} on an intermediate matrix that is usually of the largest size among all the variants, making attention evaluation expensive.
Recent Transformer models even contain millions or billions of parameters, meaning training such a model requires pretty high costs.
Fig.~\ref{fig:breakdown} demonstrates the runtime breakdown of secure LeViT-256 inference and an attention block running on \code{Semi2k}.
In various Transformer models, attention dominates the evaluation time of LeViT-256 inference.
For example, \cite{li2022mpcformer} also shows that \code{Softmax} and matrix multiplication dominate the runtime of $BERT_{BASE}$, indicating the domination of attention.
Thus, it is essential to optimize the evaluation of attention blocks.

Existing works such as \cite{dong2023puma,li2022mpcformer} focus on optimizing basic operations like \code{Softmax} and \code{GeLU} without exploiting the characteristics of attention.
Contrarily, we treat an attention block as a whole and perform optimizations across operations. 

We normally compute $\code{Softmax}(X_i)$ as $\frac{e^{X_i-\code{Max}(X)}}{\sum_i e^{X_i-\code{Max}(X)}}$ in secure computation to avoid overflow and achieve numerical stablility~\cite{dong2023puma,watson2022piranha}.
We denote the step of finding and subtracting the maximum value as \code{MaxCut}, as Fig.~\ref{fig:plaintext_attention} and \ref{fig:naive_secure_attention} illustrate.
Fig.~\ref{fig:softmax} shows the detail of a \code{Softmax} circuit.

We propose a series of new optimizations for attention evaluation in secure inference.
To reduce arithmetic computation on 64-bit numbers, we temporarily convert the representation to 32-bit fixed points and convert it back for free when necessary.
To optimize exponentiation in attention, we explore how the input of exponentiation in attention differs from general exponentiation.
We also find that we can break down \code{Softmax} and integrate the operations with other operations in attention to get better performance.
We finally implement a kernel to parallelize independent matrix multiplications for a multi-head attention mechanism.
Fig.~\ref{fig:optimized_secure_attention} illustrates our optimized attention block.

\begin{figure}[tb]
  \begin{subfigure}[t]{.15\textwidth}
    \centering
    \includegraphics[width=\linewidth]{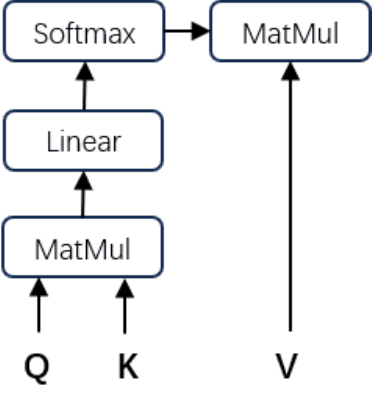}
    \vspace{-0.25in}
    \caption{Original attention}
    \label{fig:plaintext_attention}
  \end{subfigure}
  \hfill
  \begin{subfigure}[t]{.24\textwidth}
    \centering
    \includegraphics[width=\linewidth]{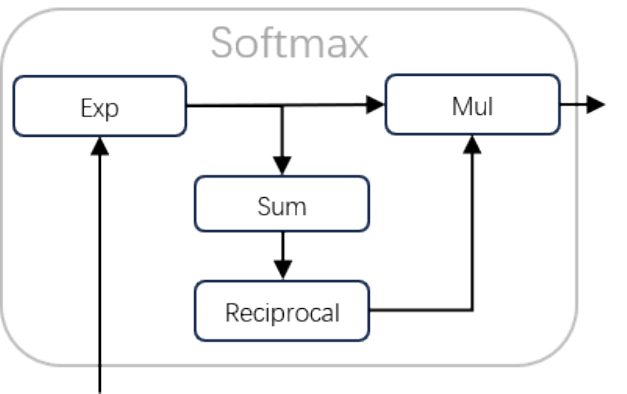}
    \vspace{-0.25in}
    \caption{Softmax}
    \label{fig:softmax}
  \end{subfigure}
  \medskip
  \medskip
  \begin{subfigure}[t]{.157\textwidth}
    \centering
    \vspace{+0.08in}
    \includegraphics[width=\linewidth]{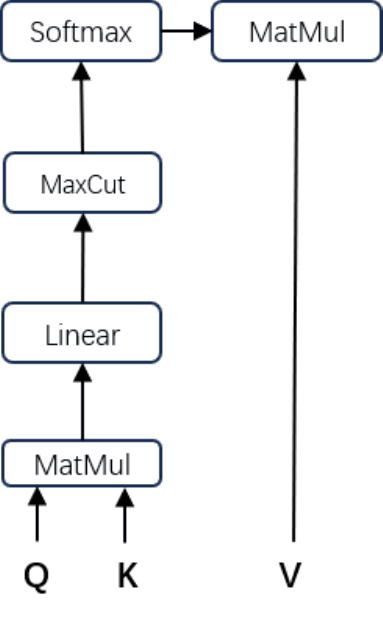}
    \vspace{-0.30in}
    \caption{Secure attention}
    \vspace{-0.30in}
    \label{fig:naive_secure_attention}
  \end{subfigure}
  \hfill
  \begin{subfigure}[t]{.24\textwidth}
    \centering
    \vspace{+0.08in}
    \includegraphics[width=\linewidth]{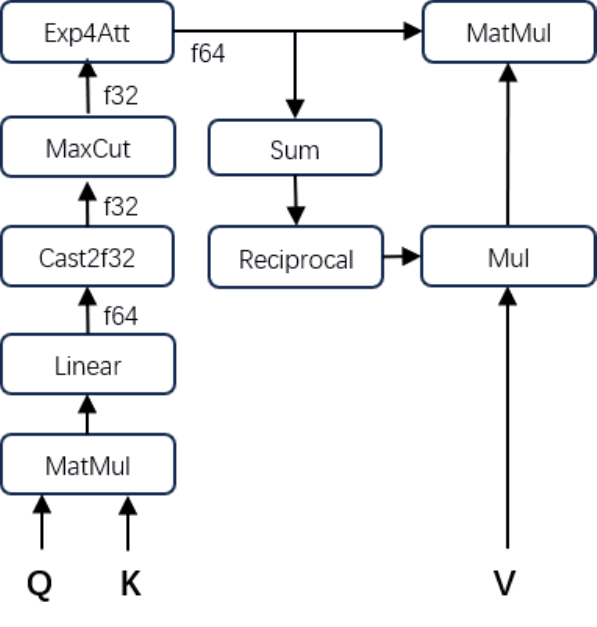}
    \vspace{-0.30in}
    \caption{Our optimized secure attention}
    \vspace{-0.30in}
    \label{fig:optimized_secure_attention}
  \end{subfigure}
  \caption{Details of attention blocks.}
  \vspace{-0.10in}
  \label{fig:attention}
  \vspace{-0.05in}
\end{figure}

\subsection{Mixing Bit Lengths}
To handle complex tasks,
\code{Fixed64} (i.e., a real number is mapped to a 64-bit fixed-point number) is necessary, as secure multiplication doubles the precision immediately, and \code{Fixed32} multiplication would overflow with the precision of more than 16 bits.
Recent works, such as \cite{watson2022piranha,knott2021crypten,dong2023puma}, usually use \code{Fixed64} throughout a task.
However, we observe that some steps in secure inference can work with 32 bits.
For example, \code{MaxCut} in attention does not require 64 bits to handle overflow so we can cast \code{Fixed64} to \code{Fixed32} before \code{MaxCut},
and our carefully designed exponentiation (see Section~\ref{section:expatt} for details) can receive \code{Fixed32} as input and outputs \code{Fixed64} without extra costs.
As casting a secret-shared value from \code{Fixed64} to \code{Fixed32} can be performed locally in each party,
we can nearly halve the cost of secure comparison that dominates the cost of \code{MaxCut}, \emph{for free}.

\subsection{Exponentiation for Attention}
\label{section:expatt}

We customize Algorithm~\ref{alg:ssne}
with several novel optimizations specific to attention.

\para{Moving scaling to plaintext.}
In secure inference, the model handler can perform any plaintext processing on model parameters beforehand.
At the attention block level, the model parameters include the weight $W$ and bias $b$ of \code{Linear} (see Fig.~\ref{fig:optimized_secure_attention}).
As the scaling step (line~\ref{exp:product}) of Algorithm~\ref{alg:ssne} uses a constant number $\log_2e$, we can actually move the scaling to \code{Linear},
that is, we use $W' = \log_2e \cdot W$ and $b' = \log_2e \cdot b$ to replace $W$ and $b$ respectively and do not need scaling in exponentiation anymore.
It is easy to see that this replacement does not hurt the correctness.
Furthermore, due to the elimination of the scaling, the input of line~\ref{exp:addp} is of precision $p$, so we do not need truncation in the next step (line~\ref{exp:decompose}).
We thus eliminate a scaling and a truncation for exponentiation.

\para{Casting from \code{Fixed32} to \code{Fixed64} for free.}
Unlike casting from \code{Fixed64} to \code{Fixed32} that can be performed locally in each party,
the normal way to cast from \code{Fixed32} to \code{Fixed64} is to extract the sign bit and extend the number with it.
However, sign bit extraction is usually a costly operation~\cite{keller2020mp,li2019privpy,mohassel2018aby3}.
In Algorithm~\ref{alg:ssne}, as there is an inherent bit decomposition (line~\ref{exp:decompose}), we only need to perform the \code{Bit2A} conversion from binary representation to 64-bit arithmetic representation as usual,
irrespective of whether the input is of bit length 32 or 64.
That is, we successfully \emph{convert \code{Fixed32} to \code{Fixed64} without extra costs}.

\para{Completing the square to save secure multiplications.}
We use Remez approximation to evaluate exponentiation on decimal fractions (line~\ref{exp:evaluate} of Algorithm~\ref{alg:ssne}).
Specifically,  we use a 4-degree polynomial (Eq.(1)), which requires three secure multiplications and four scalings.
We reduce the number of secure multiplications by \emph{completing the square}.

\begin{equation*}
\small
\begin{aligned}
f(x) &= k_4 x^4 + k_3 x^3 + k_2 x^2 + k_1 x + k_0\ \ &(1) \\
     &= k_4 (((x+t_3)^2 + t_2)^2 + t_1 x + t_0)\ \ &(2) \\
f'(x) &= k_4'(((x+t_3)^2 + t_2)^2 + t_1 x + t_0) \ \ &(3)\\
      &= k_4'((x+t_3)^2 + t_2)^2 + (k_4't_1) x + k_4't_0 &(4)
\end{aligned}
\end{equation*}

Eq.(2) shows the result of completing the square. It requires two secure multiplications and two scalings.
We further observe that $k_4$ is a common factor shared by the outputs of exponentiation on different inputs (see line~\ref{exp:evaluate} - \ref{exp:result} of Algorithm~\ref{alg:ssne}).
It will be eliminated in \code{Softmax}, thus the concrete value of $k_4$ does not affect the output of \code{Softmax}.
However, ignoring $k_4$, though saving a scaling, may make the outputs of exponentiation too large or too small, which may increase the relative error of the following steps. 
We address this issue by replacing $k_4$ with the nearest power of 2 so that the output range remains similar, then using truncation to replace the scaling.

Notably, this truncation can be free, as we can fuse it with the subsequent secure multiplications.
For example, we can replace the $k_4 = 0.0135$ of $f_1(x)$ in Algorithm~\ref{alg:ssne} with $k_4' = 2^{-6}$,
then truncate the second square $p+6$ bits instead of $p$ bits.
Finally, we get Eq.(4), which needs two secure multiplications and a scaling, \emph{saving one secure multiplication and three scalings} compared to Eq.(1).
Note that the coefficients $k_4't_0$ and $k_4't_1$ can be pre-calculated in plaintext.
Appendix~\ref{appendix:coefficients} shows the values of these newly introduced parameters.

\para{Ignoring unnecessary bits of the integer part.}
Remember that we evaluate exponentiation as $e^x = 2^{x \cdot \log_2e} = 2^{x'_{int}} \cdot f_1(x'_{dec})$, where $x'$ is $x \cdot \log_2e$, $x'_{int}$ is the integer part of $x'$,
$x'_{dec}$ is the decimal part of $x'$, and $f_1(x'_{dec})$ is the Remez approximation of $2^{x'_{dec}}$.
In Algorithm~\ref{alg:ssne}, we use a bias $p$ to ensure a correct output for $x' \in [-p, 0)$,
and line~\ref{exp:intbit} determines the maximum number of bits for calculating $2^{(x'+p)_{int}}$.
In the general case, $l = 64$ and $p = 14$, thus $c = 6$.
However, for \code{Softmax}, we have $x' \le 0$ after \code{MaxCut}, and if we adjust the evaluation of \code{Softmax}
from $\frac{e^{X_i-\code{Max}(X)}}{\sum_i e^{X_i-\code{Max}(X)}}$ to $\frac{e^{X_i-\code{Max}(X)-\epsilon}}{\sum_i e^{X_i-\code{Max}(X)-\epsilon}}$
where $\epsilon$ is a negligible number (e.g., $\epsilon = 2^{-14}$), we can get $x' < 0$, indicating $x' + p < p$.
Furthermore, since the result would underflow and would be set to $0$ when $x' < -p$, we only need to consider the case $x' + p \in [0, p)$.
Therefore, in the \code{Softmax} case, we can set $c = \log_2 p = 4$.
In this way, we further \emph{save two secure multiplications}.

Based on the above optimizations, we save 3 secure multiplications, 4 scalings, and 1 truncation for attention-specific exponentiation in total.
We summarize this variant of exponentiation in Appendix~\ref{appendix:exponentiation}.

\subsection{Downsizing Batch Multiplication.}
The last step of \code{Softmax} is to multiply $\frac{1}{\sum_i e^{X_i - \code{Max}(X_i)}}$ and $e^{X_i - \code{Max}(X_i)}$ for all $i$,
which is a batch multiplication, namely the element-wise multiplication of two arrays of the same size.
In an attention block, the size of the input of \code{Softmax} is usually larger than the size of other operations,
thus moving this batch multiplication elsewhere can downsize it.

Fig.~\ref{fig:softmax} illustrates the \code{Softmax} circuit.
The operation \code{Mul} denotes batch multiplication, and we use the symbol $*$ to denote this operation.
Now let us assume the outputs of \code{Exp}, \code{Sum} and \code{Reciprocal} in Fig.~\ref{fig:softmax} are $E$, $S$ and $R$ respectively,
then if the size of $E$ is $D_1 \times D_2$, it is easy to deduce that the size of $R$, as well as the size of $S$, is $1 \times D_2$.
Of course, we should first tile the elements of $R$ for $D_1$ times to get a new array of size $D_1 \times D_2$ before evaluating \code{Mul}.
Here we assume the notation $E * R$ implies this tiling.
Then if we put \code{Softmax} and the final matrix multiplication together, the output of an attention block can be written as $(E * R) @ V$ (see Fig.~\ref{fig:naive_secure_attention}),
where $V$ is of size $D_2 \times D_3$ and the symbol $@$ here stands for matrix multiplication.

On the other hand, we observe that $(E * R) @ V = E @ (R * V)$. 
The size of \code{MatMul} remains, whereas the size of \code{Mul} changes from $D_1 \times D_2$ to $D_2 \times D_3$.
In attention blocks, we usually have $D_1 = D_2$ and $D_2$ is multiple times larger than $D_3$.
For example, in the original paper of attention~\cite{vaswani2017attention}, we have $D_1 = D_2 = 512$ and $D_3 = 64$, so we can downsize \code{Mul} by $8\times$, as $\frac{D_1 \times D_2}{D_2 \times D_3} = 8$,
while in LeViT-256~\cite{graham2021levit}, we have $D_1 = D_2 = 196$ and $D_3 = 64$, indicating a benefit of more than $3\times$. 

\subsection{Parallelizing Matrix Multiplications}
Transformer models rely on multi-head attention to better amplify key elements and filter less important ones.
Multi-head attention mechanism involves multiple independent matrix multiplications,
and matrix multiplication in attention blocks usually works on matrices of size hundreds $\times$ hundreds.
Such a matrix multiplication can hardly fully utilize the computation resources of a GPU like A100 with thousands of cores,
as allocating too few computation tasks to each core cannot balance the time of communication, scheduling, and computation.
Thus if we calculate independent matrix multiplications in series, the GPU resources cannot be fully utilized.

We address this issue by implementing a CUDA kernel that receives two arrays of matrices as input and calculates matrix multiplications in parallel.
This operation can be formalized as $\share{\textbf{Z}}{A} \leftarrow \code{PMatMul}(\share{\textbf{X}}{A}, \share{\textbf{Y}}{A})$,
where $\textbf{X}$, $\textbf{Y}$ and $\textbf{Z}$ are arrays of matrices and $Z_i = X_i @ Y_i$.
At the implementation level, we divide the GPU cores into several groups and feed independent tasks to different groups.
Each group handles several matrix multiplications so that the GPU cores do not wait too long for scheduling and communication.
We set the number of core groups to $8$ in our paper.


%

	\section{Experiments}\label{sec:experiment}
We run four types of experiments.
The first one includes several microbenchmarks demonstrating the benefits of our optimization ideas of utilizing hardware.
The second is a function-level comparison among implementations from \cite{knott2021crypten},
state-of-the-art approaches from \cite{catrina2010secure, keller2022secure, aly2019benchmarking}, and the novel algorithms of \name.
The third is the performance of neural network training and inference, compared with Piranha~\cite{watson2022piranha}, the state-of-the-art GPU-assisted framework.
In the last one, we report the performance of secure inference of a pre-trained Transformer model with 18.9 million parameters.

\subsection{Experiment Setup}
Each party uses a machine with 512 GB RAM and uses AMD EPYC 7H12 CPU (2.6 GHz, 64 cores with 128 threads).
The engine of each party runs on an Nvidia A100 GPU with 80 GB memory and uses a Mellanox ConnectX-6 for networking.
Both A100 and ConnectX-6 interact with the host through PCIe 4.0.
The machines are connected by a Mellanox Quantum QM8790 200 Gbps InfiniBand switch.
With the QM8790 and ConnectX-6, the parties can communicate using both socket and RDMA, achieving bandwidths of up to 25 Gbps for socket channels and 200 Gbps for RDMA.

We implement \name using C++ 17 and use CUDA toolkit 11.6 for GPU-related operations.
We offer a flag that allows users to choose between CPU, GPU, and hybrid mode.
We use the RDMA-core library~\cite{rdma-core} to implement our RDMA-based network components.
In total, there are 31k lines of C++ code, including all components and algorithms.

\subsection{Microbenchmarks}
\para{Integer kernel.}
We compare the performance of our 64-bit integer kernel for matrix multiplication in \name with Piranha's CUTLASS implementation  (\url{https://github.com/jlwatson/piranha-cutlass}, commit id \code{e4f6c42}).
We employ an Nvidia A100 to run plaintext matrix multiplication of different matrix sizes. Fig.~\ref{fig:gemm} shows that for matrix size smaller than $1000 \times 1000$, \name is about 2.4$\times$ faster than Piranha-CUTLASS,
and for larger matrices, the performance of \name is close to Piranha-CUTLASS.

\begin{figure*}[tb]
\centering
\subfloat[\label{fig:gemm}]{\includegraphics[width=0.3\textwidth]{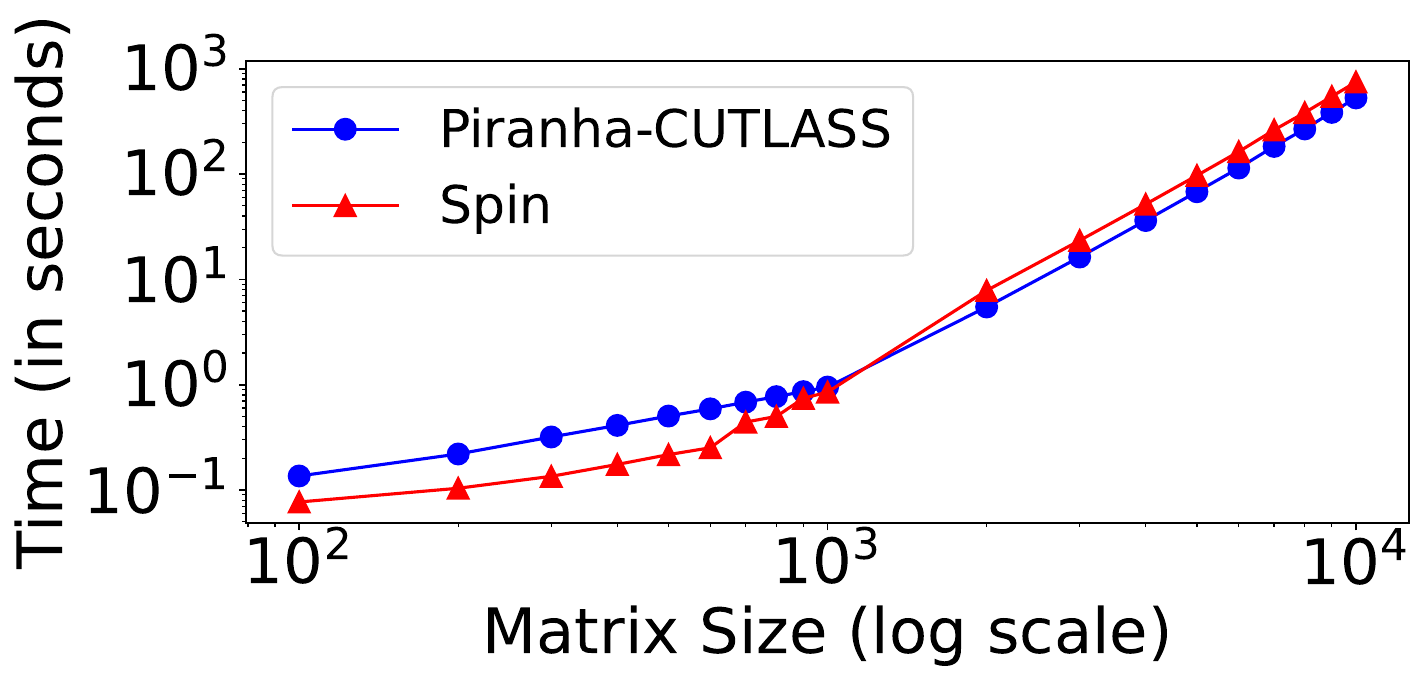}\vspace{-0.10in}}
\hfill
\subfloat[\label{fig:channel_rdma}]{\includegraphics[width=0.3\textwidth]{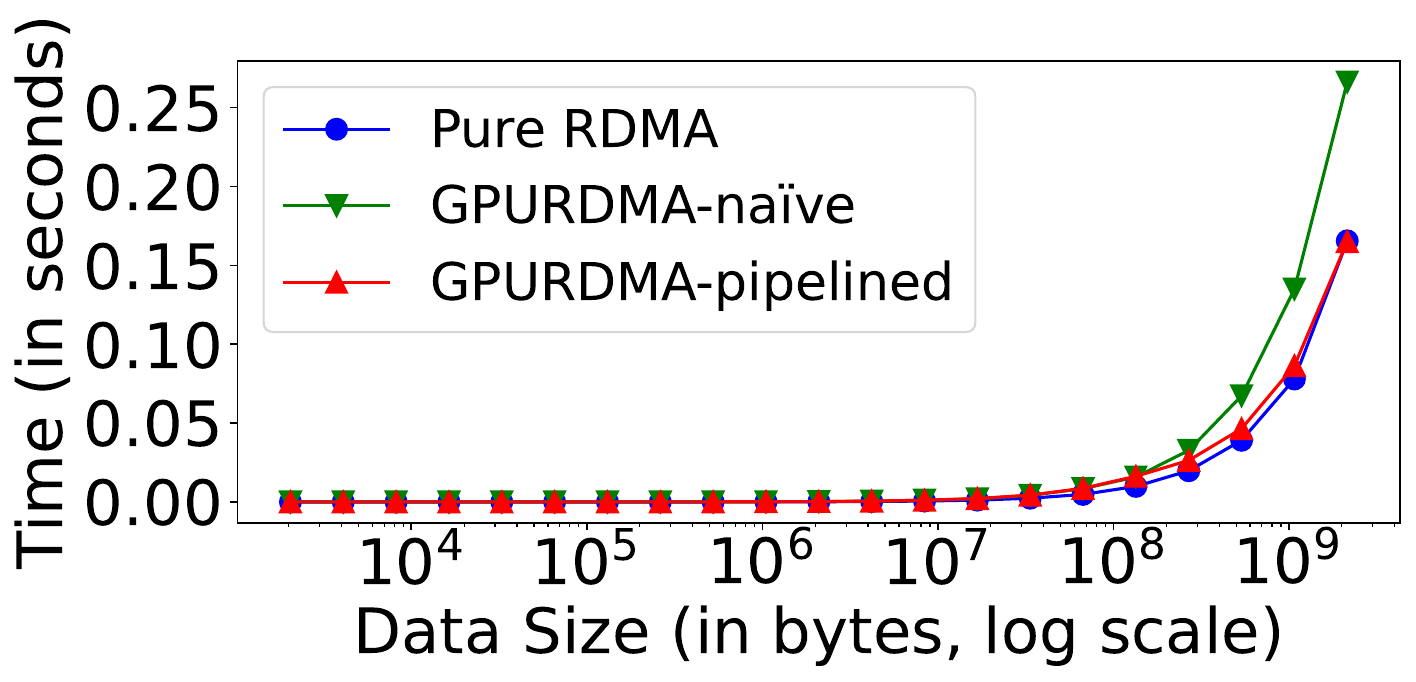}\vspace{-0.10in}}
\hfill
\subfloat[\label{fig:channel_socket}]{\includegraphics[width=0.3\textwidth]{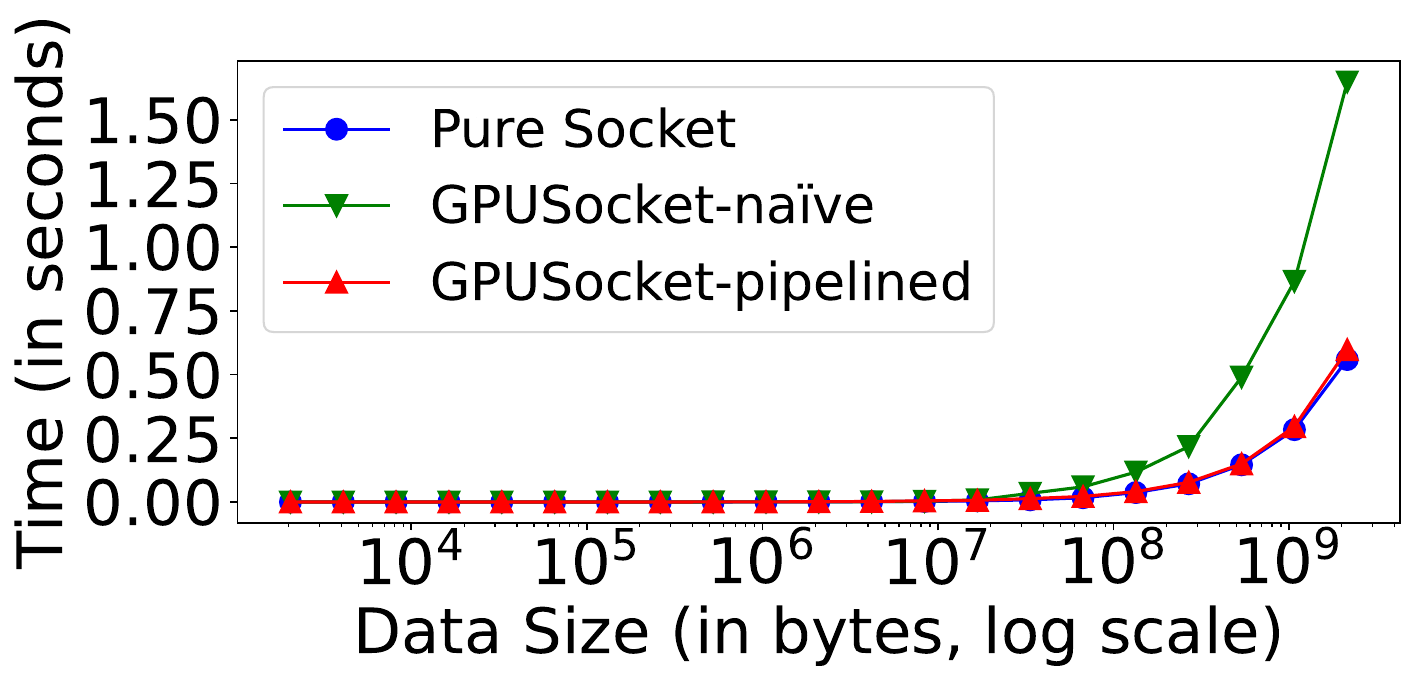}\vspace{-0.10in}}
\vspace{-0.10in}
\caption{The performance of our integer kernel and pipelined communication solution.}
\vspace{-0.10in}
\label{fig:channel}
\end{figure*}

\para{Pipelined data transfer.}
As mentioned in Section~\ref{sec:comm}, we use a \emph{network agent} on each host to pipeline the data transfer so that the data copy between hosts and GPUs would not hinder the end-to-end performance.
We record the latency of transferring data of different sizes from a GPU to a remote one, using RDMA and socket, respectively.
For comparison, we also record the latency of pure RDMA/socket channels (i.e., the time for transferring data with pure RDMA/socket channels, without the latency of data copy between GPUs and hosts), and the latency of the na\"ive method mentioned in Section~\ref{sec:comm}.
As Fig.~\ref{fig:channel_rdma} and Fig.~\ref{fig:channel_socket} show, our pipelined solution nearly eliminates the influence of data copy between hosts and GPUs, while achieving about $2\times$ performance improvement compared to the na\"ive method.

\subsection{Non-linear Functions}
\label{sec:non_linear_eval}
We conduct benchmarks on three non-linear functions: reciprocal, exponentiation, and logarithm.
We compare our algorithms to both na\"ive and state-of-the-art (SOTA) approaches.
To our knowledge, the SOTAs are:
reciprocal in \cite{catrina2010secure}, exponentiation in \cite{keller2022secure}, and logarithm in \cite{aly2019benchmarking}.
We refer to them as \emph{SOTA-rec, SOTA-exp, SOTA-log}, respectively.
Crypten~\cite{knott2021crypten} presents \emph{na\"ive} approaches for these functions:
Newton-Raphson method for reciprocal, limit approximation for exponentiation, and Householder method for logarithm.

Fig.~\ref{fig:nonlinear_accuracy_time} shows the comparison of accuracy and running time.
The left column records the accuracy, while the right column is for running time.
The parameters are all from the default settings in the articles or related codes.
Compared with the na\"ive implementations, both our algorithms and SOTAs have remarkable advantages in accuracy.
Compared to SOTAs, our algorithms for exponentiation and logarithm are more accurate and stable.
For $x > 10$, SOTA-exp even overflows.
Note that our reciprocal algorithm is based on SOTA-rec with several optimizations so that the accuracy curves overlap.


As the accuracy of the na\"ive implementations is not practical enough, we only compare the running time of our algorithms with SOTAs.
We run them on \name backend in a two-party setting.
Each party uses 20 CPU cores for computation, and the two parties communicate using RDMA.
We run the evaluation 10 times and take the average running time.
Compared to SOTAs, the advantage of our algorithms increases as the input scale grows.
When the input scale reaches $10^6$ elements, the advantage of our algorithms can be up to about $15\%$ on reciprocal/exponentiation and $100\%$ on logarithm.

\begin{figure}[tb]
\centering
\subfloat[\label{fig:reca}]{\includegraphics[width=0.22\textwidth]{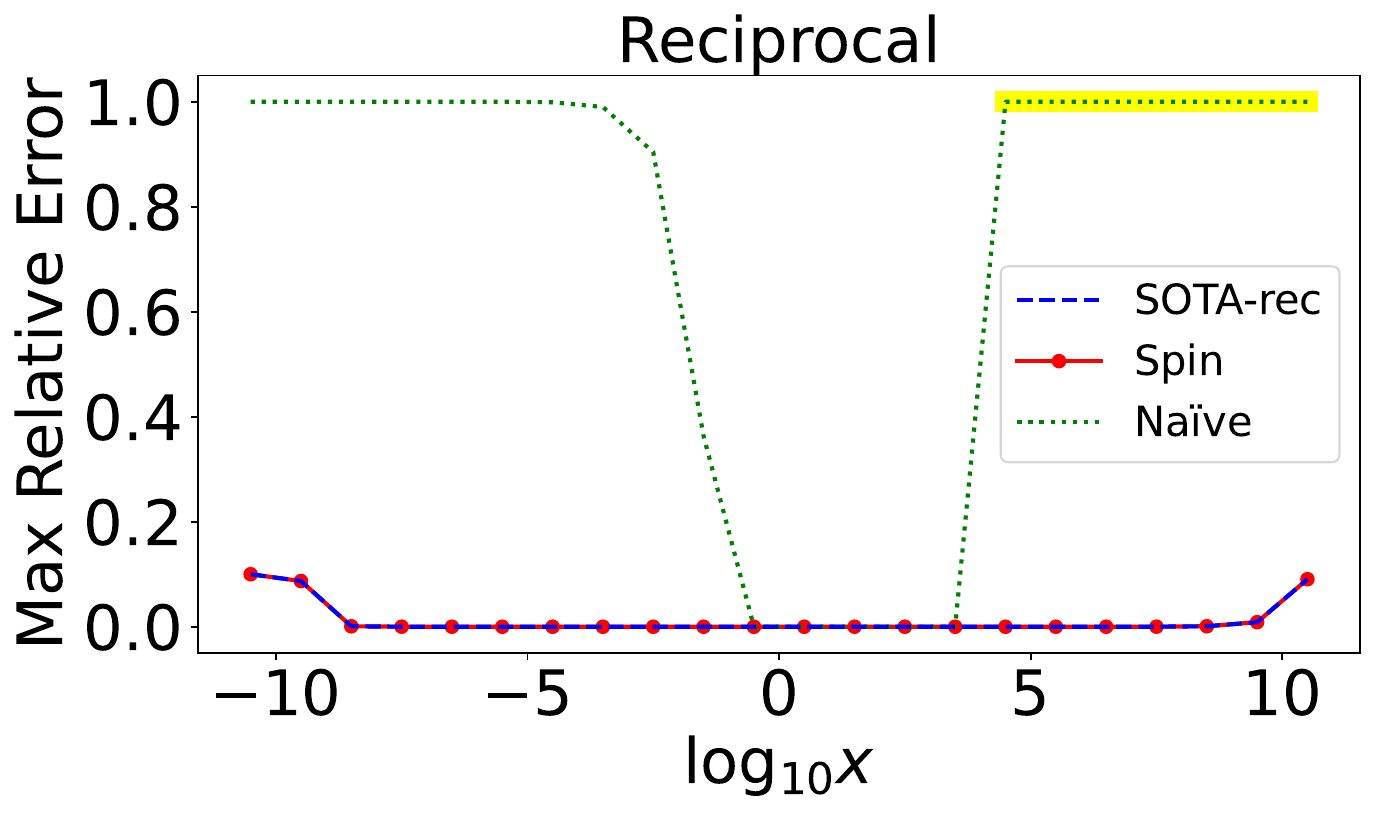}\vspace{-0.10in}}
\hfill
\subfloat[\label{fig:rect}]{\includegraphics[width=0.21\textwidth]{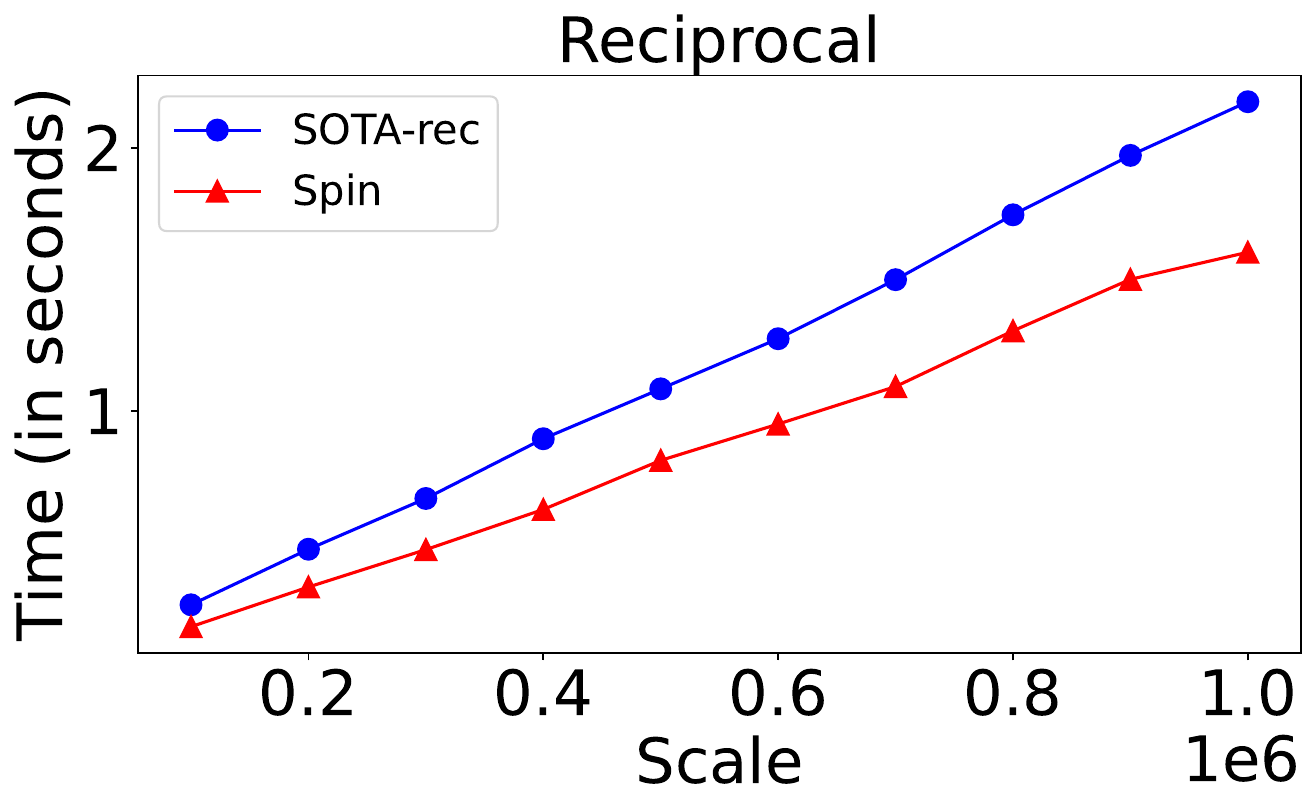}\vspace{-0.10in}}
\hfill
\subfloat[\label{fig:expa}]{\includegraphics[width=0.22\textwidth]{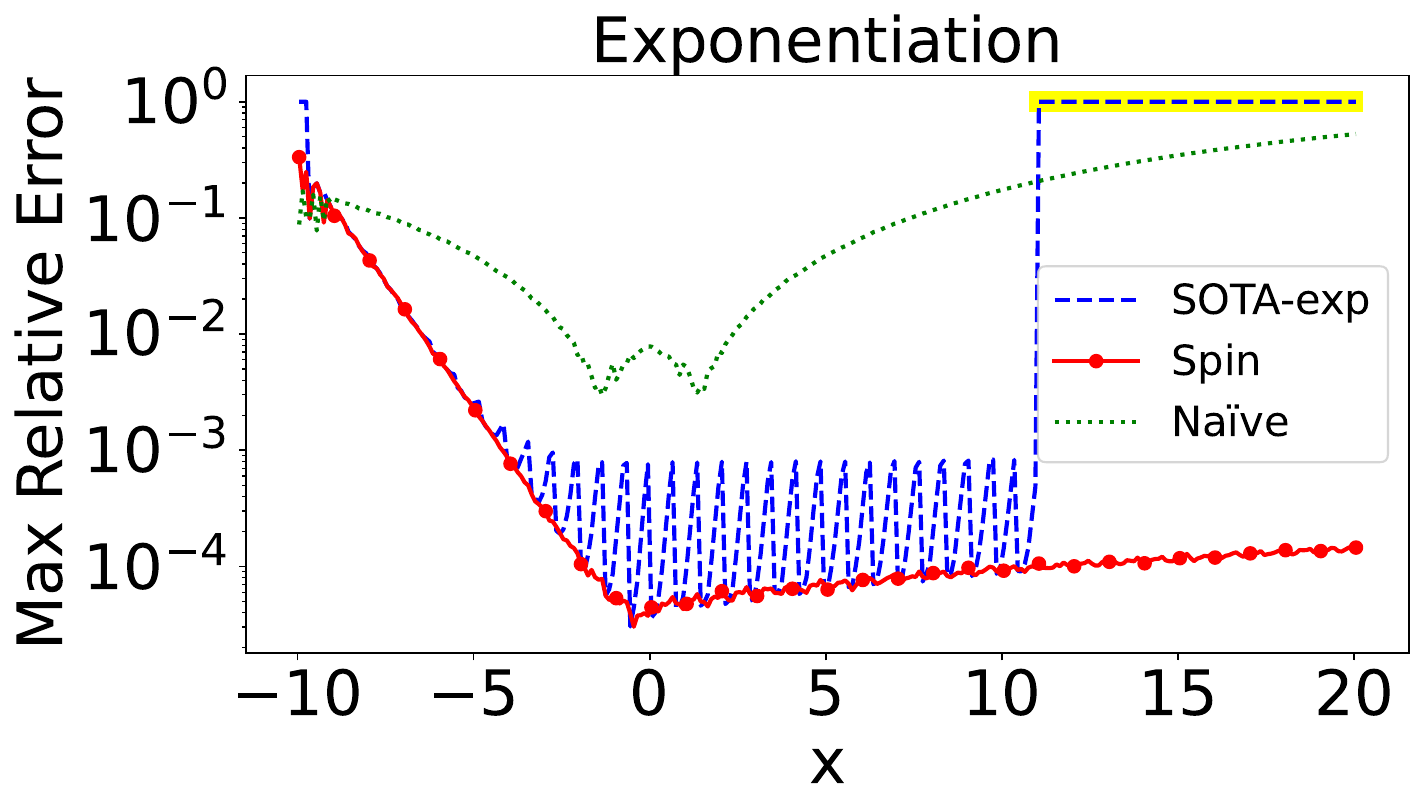}\vspace{-0.10in}}
\hfill
\subfloat[\label{fig:expt}]{\includegraphics[width=0.21\textwidth]{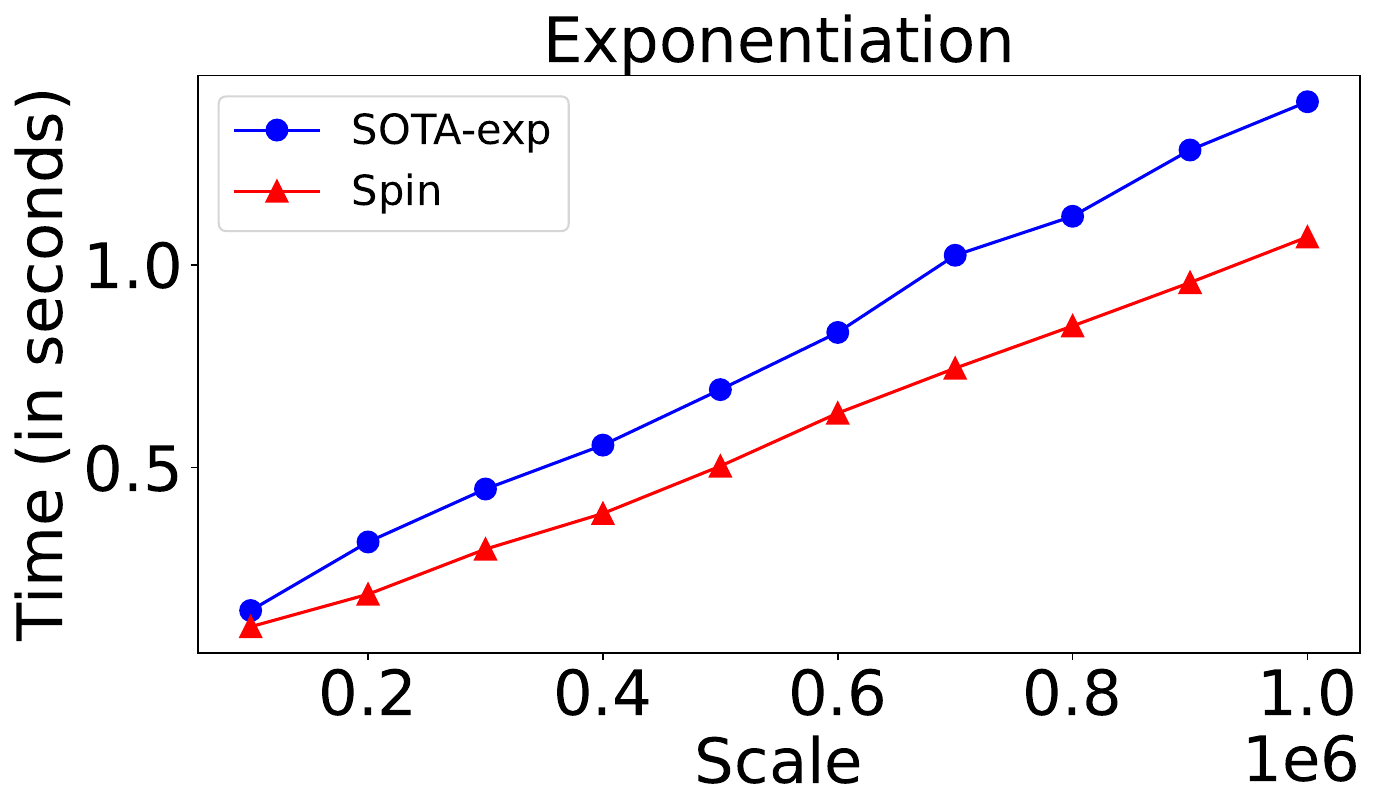}\vspace{-0.10in}}
\hfill
\subfloat[\label{fig:loga}]{\includegraphics[width=0.22\textwidth]{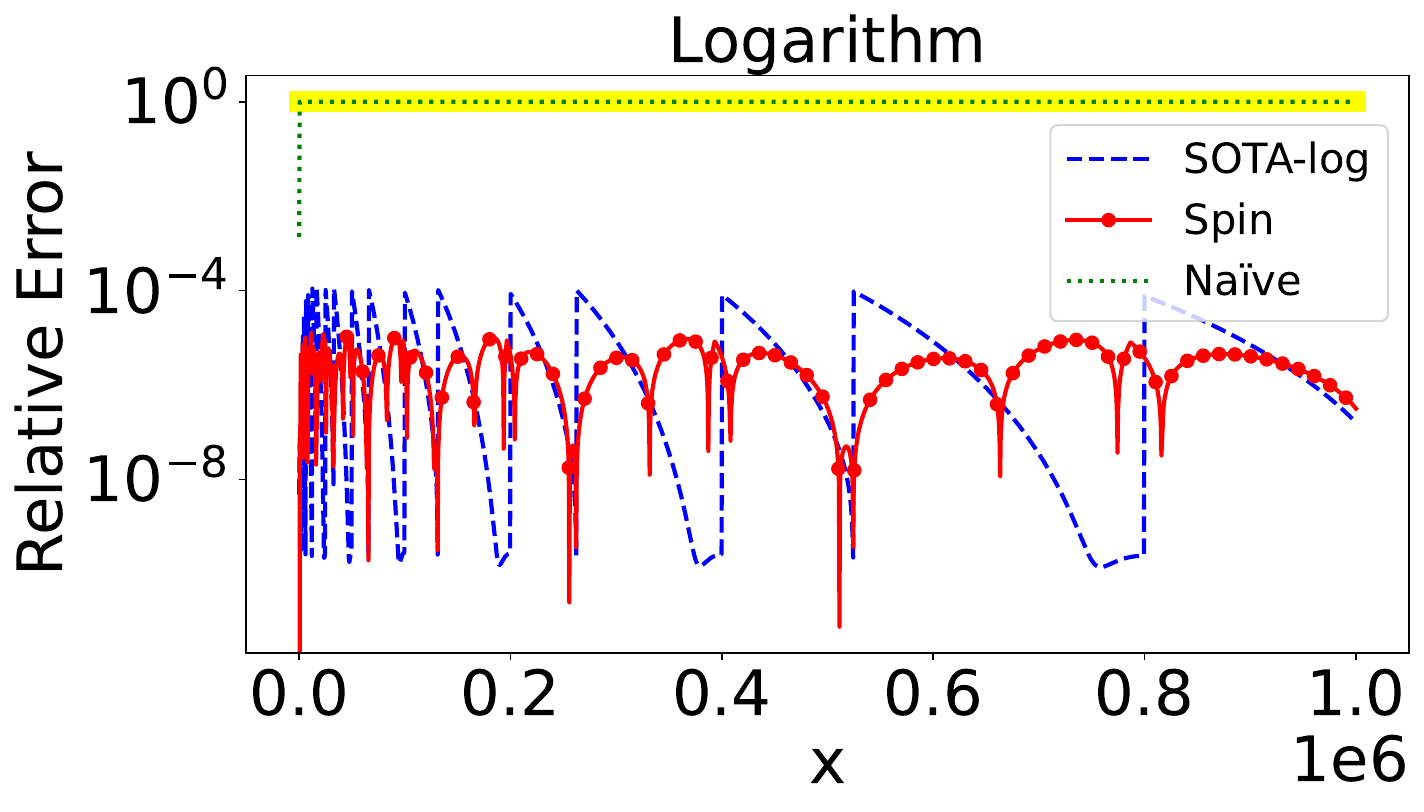}\vspace{-0.10in}}
\hfill
\subfloat[\label{fig:logt}]{\includegraphics[width=0.205\textwidth]{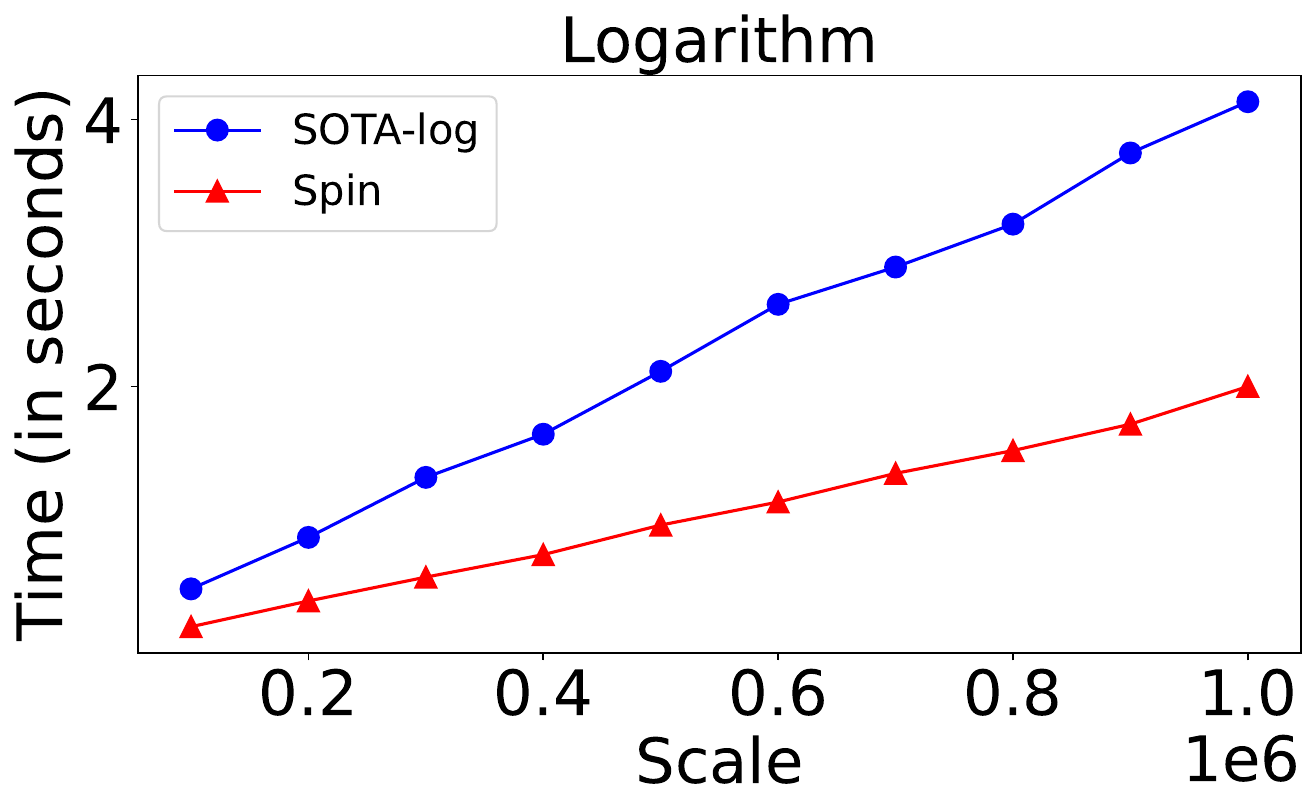}\vspace{-0.10in}}
\vspace{-0.10in}
\caption{Accuracy and running time of non-linear functions.
         The segments colored yellow mean overflows.
         } 
\vspace{-0.10in}
\label{fig:nonlinear_accuracy_time}
\end{figure}


\subsection{Neural Network Training and Inference}
We compare our empirical performance on real neural networks with Piranha~\cite{watson2022piranha}, the framework with state-of-the-art performance, to the best of our knowledge.
We compare \name in a 2-party setting with two protocols from Piranha: \emph{P-SecureML}, a 2-party protocol, and \emph{P-Falcon}, which is based on a 3-party protocol with the assumption of honest majority.

We use the same set of neural network models utilized in Piranha's evaluation~\cite{watson2022piranha}, including:
a Simple model from SecureML~\cite{mohassel2017secureml}, LeNet~\cite{lecun1998gradient}, AlexNet~\cite{krizhevsky2017imagenet}, VGG16~\cite{simonyan2014very}.
The first two models use MNIST~\cite{lecun1998mnist} as the dataset, and others use CIFAR10~\cite{krizhevsky2009learning}.
MNIST has 60,000 / 10,000 training/test samples, each of which has a size of $28\times 28$. CIFAR10 has 50,000 / 10,000 samples, each of which is $32\times 32$ with RGB channels.
Following Piranha's setting, we modify the models by replacing the max-pooling layer with the average pooling layer.
In our experiments, we remove a few layers in VGG16, as we observe that these layers have minimal impact on the final accuracy. We denote the modified version as \emph{VGG16m}.
We use the same Kaiming initialization~\cite{he2015delving}, learning rate, and mini-batch size as Piranha.
We list the detailed structure for all neural networks in Appendix~\ref{networks}.

We run Piranha and \name in the same testbed mentioned above.
The built-in network module of Piranha only supports sockets over TCP/IP.
To make a fair comparison, we further optimize the socket with pipelining and implement an RDMA-based network module for Piranha.
\name can run in three modes: CPU, GPU, and hybrid mode.
Both CPU mode and GPU mode can run with sockets or RDMA, while the hybrid mode can use both (as Section~\ref{sec:hybrid} details).


\para{Neural network training time.}
Table~\ref{eva:training_time} displays the running time of one \emph{epoch} (one epoch consists of training on all data samples once).
Consistent with Piranha's setting, we only report online time here.
We have the following observations:
\begin{itemize}[leftmargin=0.4cm]
    \item RDMA in the Simple model is slower.  This is because the computation time is so small that the overhead of moving data in RDMA dominates the total time.
    \item For large neural networks, RDMA has clear advantages in all three protocols. 
    \item For all the models other than the Simple model, the hybrid mode of \name outperforms all the other settings, achieving up to $2\times$ faster than the RDMA-enhanced P-Falcon.
\end{itemize}

As a side note, \name offers stronger security than P-Falcon. 
Piranha computes reciprocal in plaintext to accelerate \code{Softmax} for accuracy and efficiency.
The code of Piranha offers a fully secure \code{softmax} version (\url{https://github.com/ucbrise/piranha}, commit id \code{dfbcb59}),
using a quadratic polynomial to approximate reciprocal, but the training fails (the accuracy converges to about 10\%) when using this rough approximation.  
That is, even with stronger security, \name-hybrid still outperforms P-Falcon.



\begin{table}[tb]
\footnotesize
\centering
\renewcommand{\arraystretch}{1.05}
\begin{tabular}{ |c|c|c|c|c|c| } 
    \hline
    & Channel & Simple & LeNet & AlexNet & VGG16m \\ 
    \hline
    \multirow{2}{*}{\name-CPU} & socket & \textbf{26.6} & 539.6 & 552.7 & 21160.6 \\
        & RDMA & 34.0 & 510.5 & 546.7 & 20651.1 \\
    \hline
    \multirow{2}{*}{\name-GPU} & socket & 57.0 & 187.1 & 188.8 & 1432.8 \\ 
        & RDMA & 59.1 & 158.0 & 141.5 & 835.3 \\
    \hline
        \name-hybrid & hybrid & 32.0 & \textbf{120.5} & \textbf{119.4} & \textbf{804.4} \\
    \hline
    \multirow{2}{*}{P-SecureML} & socket & 65.8 & 271.9 & 385.0 & 1847.3 \\
        & RDMA & 73.8 & 185.7 & 274.3 & 1442.5 \\
    \hline
    \multirow{2}{*}{P-Falcon} & socket & 48.0 & 266.4 & 345.6 & 1692.4 \\
        & RDMA & 49.0 & 162.9 & 265.8 & 1686.8 \\
        \hline
\end{tabular}
\vspace{-0.10in}
\caption{Neural network training time for one epoch.}
\vspace{-0.05in}
\label{eva:training_time}
\end{table}

\para{Accuracy of trained models.}
The accuracy results are obtained by training 10 epochs on the training set and evaluating the test set.
The precisions (i.e., the precision parameter $p$ as defined in Section~\ref{sec:non_linear}) are 23/23/23/26 bits for Simple/LeNet/AlexNet/VGG16m.
We train the models from scratch by initialing all parameters randomly, to better demonstrate \name's capabilities.
Table~\ref{eva:accuracy} demonstrates that \name produces better accuracy for all neural networks. 
This is mainly thanks to our optimized algorithms for non-linear functions, compared to Piranha's simple polynomial approximation.  

We also train and test the model in a plaintext version using double-precision floating points,
and the result shows that \name achieves results \emph{closed to the plaintext version}, while P-SecureML and P-Falcon result in much lower accuracy in AlexNet.
In total, we can see that \name outperforms Piranha in both efficiency and accuracy.

\begin{table}[tb]
\footnotesize
\centering
\begin{tabular}{ |c|c|c|c|c|c| } 
    \hline
    & DataType & Simple & LeNet & AlexNet & VGG16m \\ 
    \hline
    \name & Fixed64 & \textbf{97.48} & \textbf{98.89} & \textbf{54.79} & \textbf{63.17} \\
    \hline
    P-SecureML & Fixed64 & 96.56 & 97.13 & 37.01 & 53.08 \\
    \hline
    P-Falcon & Fixed64 & 97.10 & 97.13 & 38.97 & 53.08 \\
    \hline
    Plaintext & \texttt{double} & 97.83 & 99.08 & 55.16 & 64.41 \\
    \hline
\end{tabular}
\vspace{-0.10in}
\caption{Neural network accuracy on the test set.}
\vspace{-0.10in}
\label{eva:accuracy}
\end{table}

\para{Inference performance. }
We also display the inference performance (measured by images per second) for all neural networks in Table~\ref{infer}.
We run Piranha using our RDMA implementation, and we set \name to operate in the hybrid mode.
The results show an improvement up to $2.6\times$ over P-Falcon. 

\begin{table}[tb]
\footnotesize
\centering
\begin{tabular}{ |c|c|c|c|c| } 
    \hline
    & Simple & LeNet & AlexNet & VGG16m \\ 
    \hline
    \name & \textbf{13485.81} & \textbf{3288.98} & \textbf{2999.40} & \textbf{166.30} \\
    \hline
    P-SecureML & 12241.10 & 1598.32 & 1472.42 & 44.65 \\
    \hline
    P-Falcon & 12025.87 & 2149.20 & 1767.94 & 62.88 \\
    \hline
\end{tabular}
\vspace{-0.10in}
\caption{Inference performance (images per second).}
\vspace{-0.10in}
\label{infer}
\end{table}


\para{Scalablility with number of parties}
\name can theoretically work on any number of parties with a dishonest majority.
We evaluate the performance of \name over the Simple model with 2-5 parties, using RDMA for communication.  
Table~\ref{eva:scale} shows the results, where $Rate = \frac{\text{running time}}{\text{2-party running time}}$.
As the number of parties increases, the communication complexity of online computation grows linearly.
We observe a stable running time that increases linearly with the number of parties.

\begin{table}[!htb]
\footnotesize
\centering
\begin{tabular}{ |c|c|c|c|c| } 
    \hline
    Parties & 2 & 3 & 4 & 5 \\ 
    \hline
    Time & 37.24 & 54.90 & 82.22 & 102.88 \\
    \hline
    Rate & 1.00 & 1.47 & 2.21 & 2.76 \\
    \hline
\end{tabular}
\vspace{-0.10in}
\caption{Scalability of \name with different numbers of parties.}
\vspace{-0.10in}
\label{eva:scale}
\end{table}

\subsection{Transformer Inference}
The goal of this experiment is to demonstrate the benefits of our novel optimizations in Section~\ref{sec:attention}.
We first test the running time and communication incurred by the first attention block of LeViT-256 with the whole validation set and additionally collect \emph{mean squared error (MSE)} values.
This evaluation can show the performance of different solutions on an attention block with the same input.
We then use LeViT-256 inference as an example of Transformer models and record the throughput and accuracy.
LeViT~\cite{graham2021levit} is a Vision Transformer for fast inference of image classification.
We choose LeViT-256 with 18.9 million parameters and use the ImageNet~\cite{deng2009imagenet} dataset, which has images of size $224\times 224$ with RGB channels.

\para{Performance of \name.}
Piranha, the state-of-the-art GPU-assisted MPC framework, has not supported Transformer models yet.
MPCFormer~\cite{li2022mpcformer} is one of the latest frameworks for secure Transformer inference.
An important issue of MPCFormer is that it does not support pre-trained models, and requires models to be retrained using \emph{2Quad approximation} for non-linear functions.
PUMA~\cite{dong2023puma}, to our knowledge, is the state-of-the-art of secure inference of large pre-trained models.
The key contributions of PUMA are new approximations for \code{GeLU}, \code{Softmax}, and \code{LayerNorm}.
However, it currently only supports running on CPUs.
So we only report the performance of \name, in line 1 of Table~\ref{table:levit}.

\para{Ablation study.}
We demonstrate the benefits of our attention-specific optimizations by replacing the attention implementation with other approaches:
\begin{itemize}[leftmargin=0.4cm]
    \item \para{\name w/P-softmax}: \name backend with all the optimizations except those in Section~\ref{sec:attention}, while \code{Softmax} is replaced with PUMA's.
    \item \para{\name w/S-nonlinear}: \name backend with all the optimizations except those in Section~\ref{sec:attention}, while the non-linear functions are implemented using SOTAs as in Section~\ref{sec:non_linear_eval}.
\end{itemize}
A notable thing is that,
LeViT uses \code{Hardswish} to replace \code{GeLU} and fuses \code{LayerNorm} to \code{Linear}, so we only borrow \code{Softmax} from PUMA.
But as both Fig.~\ref{fig:breakdown} and \cite{li2022mpcformer} show, \code{Softmax} and \code{MatMul} dominate the performance of Transformer inference,
hence replacing \code{Softmax} solely is enough to demonstrate the advantages of our optimizations compared to PUMA.
The results are reported in line 2-3 of Table~\ref{table:levit}.

For attention evaluation, we can see that our solution outperforms the others in both efficiency and communication cost while achieving the best precision.
For secure inference, \name also gets the best efficiency and better accuracy than PUMA's \code{Softmax}.
In conclusion, our attention-specific optimizations enable \name to achieve advantages in efficiency, communication, and accuracy.

\begin{table}[tb]
    \footnotesize
    \centering
    \begin{tabular}{ |c|c|c|c|c|c| }
        \hline
        \multirow{2}*{} & \multicolumn{3}{|c|}{Attention Block} & \multicolumn{2}{|c|}{LeViT-256} \\
        \cline{2-6}
        & Time & Comm. & MSE & Thro. & Accu. \\ 
        \hline
        \name & \textbf{23.50} & \textbf{56.04} & \textbf{4.10e-6} & \textbf{3.08} & \textbf{80.0} \\
        \hline
        \name w/P-softmax & 28.85 & 61.43 & 4.82e-3 & 2.57 & 79.2 \\
        \hline
        \name w/S-nonlinear & 34.20 & 82.06 & 2.72e-5 & 2.38 & \textbf{80.0} \\
        \hline
    \end{tabular}
    \vspace{-0.10in}
    \caption{Performance on attention blocks and LeViT-256 inference (time in ms, communication in MB, throughput in img/s).
             The accuracy of plaintext inference is 81.6\%.}
    \vspace{-0.10in}
    \label{table:levit}
\end{table}

	\section{Conclusions and Future Work}

Implementing machine learning on MPC frameworks necessitates numerous subtle optimizations to enhance both efficiency and accuracy.
At the MPC protocol level, it is crucial to trade-off between bits utilization and the associated costs.
At the neural network level, it is helpful to exploit the characteristics of high-level components instead of basic operations.
Additionally, new generation of hardware and networks, including GPUs and RDMA, significantly speed up computation.
However, managing memory and data movements among these devices requires considerable effort.

Our optimizations in \name prove to be highly effective.
Compared to existing highly-optimized solutions, \name simultaneously achieves superior security, performance, and accuracy.
\name paves the way for several future advancements.
Firstly, we plan to utilize a GPU cluster for each party and design heterogeneous computing scheduling specific to MPC tasks.
Additionally, we aim to support larger inputs exceeding GPU memory capacity by transparently caching data in host memory and disks, which is vital for accommodating larger models.
Lastly, we will explore more optimizations for secure training of large models.

	\bibliographystyle{plain}
	\bibliography{literature}

\begin{thebibliography}{10}

\bibitem{abadi2016deep}
Martin Abadi, Andy Chu, Ian Goodfellow, H~Brendan McMahan, Ilya Mironov, Kunal
  Talwar, and Li~Zhang.
\newblock Deep learning with differential privacy.
\newblock In {\em Proceedings of the 2016 ACM SIGSAC conference on computer and
  communications security}, pages 308--318, 2016.

\bibitem{aly2019benchmarking}
Abdelrahaman Aly and Nigel~P Smart.
\newblock Benchmarking privacy preserving scientific operations.
\newblock In {\em Applied Cryptography and Network Security: 17th International
  Conference, ACNS 2019, Bogota, Colombia, June 5--7, 2019, Proceedings}, pages
  509--529. Springer, 2019.

\bibitem{beaver1992efficient}
Donald Beaver.
\newblock Efficient multiparty protocols using circuit randomization.
\newblock In {\em Advances in Cryptology—CRYPTO’91: Proceedings 11}, pages
  420--432. Springer, 1992.

\bibitem{beaver2001precomputing}
Donald Beaver.
\newblock Precomputing oblivious transfer.
\newblock In {\em Advances in Cryptology—CRYPT0’95: 15th Annual
  International Cryptology Conference Santa Barbara, California, USA, August
  27--31, 1995 Proceedings}, pages 97--109. Springer, 2001.

\bibitem{ben1988completeness}
Michael Ben-Or, Shafi Goldwasser, and Avi Wigderson.
\newblock Completeness theorems for non-cryptographic fault-tolerant
  distributed computation.
\newblock In {\em Proceedings of the twentieth annual ACM symposium on Theory
  of computing}, pages 1--10, 1988.

\bibitem{brown2020language}
Tom Brown, Benjamin Mann, Nick Ryder, Melanie Subbiah, Jared~D Kaplan, Prafulla
  Dhariwal, Arvind Neelakantan, Pranav Shyam, Girish Sastry, Amanda Askell,
  et~al.
\newblock Language models are few-shot learners.
\newblock {\em Advances in neural information processing systems},
  33:1877--1901, 2020.

\bibitem{byali2019flash}
Megha Byali, Harsh Chaudhari, Arpita Patra, and Ajith Suresh.
\newblock Flash: fast and robust framework for privacy-preserving machine
  learning.
\newblock {\em Cryptology ePrint Archive}, 2019.

\bibitem{canetti2001universally}
Ran Canetti.
\newblock Universally composable security: A new paradigm for cryptographic
  protocols.
\newblock In {\em Proceedings 42nd IEEE Symposium on Foundations of Computer
  Science}, pages 136--145. IEEE, 2001.

\bibitem{catrina2010improved}
Octavian Catrina and Sebastiaan De~Hoogh.
\newblock Improved primitives for secure multiparty integer computation.
\newblock In {\em Security and Cryptography for Networks: 7th International
  Conference, SCN 2010, Amalfi, Italy, September 13-15, 2010. Proceedings 7},
  pages 182--199. Springer, 2010.

\bibitem{catrina2010secure}
Octavian Catrina and Amitabh Saxena.
\newblock Secure computation with fixed-point numbers.
\newblock In {\em Financial Cryptography and Data Security: 14th International
  Conference, FC 2010, Tenerife, Canary Islands, January 25-28, 2010, Revised
  Selected Papers 14}, pages 35--50. Springer, 2010.

\bibitem{chaudhari2019trident}
Harsh Chaudhari, Rahul Rachuri, and Ajith Suresh.
\newblock Trident: Efficient 4pc framework for privacy preserving machine
  learning.
\newblock {\em arXiv preprint arXiv:1912.02631}, 2019.

\bibitem{CUTLASS}
{CUTLASS contributors}.
\newblock {CUTLASS}: Cuda templates for linear algebra subroutines.
\newblock \url{https://github.com/NVIDIA/cutlass}, 2023.

\bibitem{dalskov2020secure}
Anders Dalskov, Daniel Escudero, and Marcel Keller.
\newblock Secure evaluation of quantized neural networks.
\newblock {\em Proceedings on Privacy Enhancing Technologies}, 4:355--375,
  2020.

\bibitem{dalskov2021fantastic}
Anders~PK Dalskov, Daniel Escudero, and Marcel Keller.
\newblock Fantastic four: Honest-majority four-party secure computation with
  malicious security.
\newblock In {\em USENIX Security Symposium}, pages 2183--2200, 2021.

\bibitem{damgaard2003universally}
Ivan Damg{\aa}rd and Jesper~Buus Nielsen.
\newblock Universally composable efficient multiparty computation from
  threshold homomorphic encryption.
\newblock In {\em Advances in Cryptology-CRYPTO 2003: 23rd Annual International
  Cryptology Conference, Santa Barbara, California, USA, August 17-21, 2003.
  Proceedings 23}, pages 247--264. Springer, 2003.

\bibitem{fastlog}
{David Goldberg}.
\newblock Fast approximate logarithms, part iii: The formulas, 2015.
\newblock
  \url{https://tech.ebayinc.com/engineering/fast-approximate-logarithms-part-iii-the-formulas/},
  Last accessed on 2023-03-01.

\bibitem{demmler2015aby}
Daniel Demmler, Thomas Schneider, and Michael Zohner.
\newblock Aby-a framework for efficient mixed-protocol secure two-party
  computation.
\newblock In {\em NDSS}, 2015.

\bibitem{deng2009imagenet}
Jia Deng, Wei Dong, Richard Socher, Li-Jia Li, Kai Li, and Li~Fei-Fei.
\newblock Imagenet: A large-scale hierarchical image database.
\newblock In {\em 2009 IEEE conference on computer vision and pattern
  recognition}, pages 248--255. Ieee, 2009.

\bibitem{devlin2018bert}
Jacob Devlin, Ming-Wei Chang, Kenton Lee, and Kristina Toutanova.
\newblock Bert: Pre-training of deep bidirectional transformers for language
  understanding.
\newblock {\em arXiv preprint arXiv:1810.04805}, 2018.

\bibitem{dong2023puma}
Ye~Dong, Wen-jie Lu, Yancheng Zheng, Haoqi Wu, Derun Zhao, Jin Tan, Zhicong
  Huang, Cheng Hong, Tao Wei, and Wenguang Cheng.
\newblock Puma: Secure inference of llama-7b in five minutes.
\newblock {\em arXiv preprint arXiv:2307.12533}, 2023.

\bibitem{dwork2014algorithmic}
Cynthia Dwork, Aaron Roth, et~al.
\newblock The algorithmic foundations of differential privacy.
\newblock {\em Foundations and Trends{\textregistered} in Theoretical Computer
  Science}, 9(3--4):211--407, 2014.

\bibitem{escudero2020improved}
Daniel Escudero, Satrajit Ghosh, Marcel Keller, Rahul Rachuri, and Peter
  Scholl.
\newblock Improved primitives for mpc over mixed arithmetic-binary circuits.
\newblock In {\em Advances in Cryptology--CRYPTO 2020: 40th Annual
  International Cryptology Conference, CRYPTO 2020, Santa Barbara, CA, USA,
  August 17--21, 2020, Proceedings, Part II 40}, pages 823--852. Springer,
  2020.

\bibitem{fan2022nfgen}
Xiaoyu Fan, Kun Chen, Guosai Wang, Mingchun Zhuang, Yi~Li, and Wei Xu.
\newblock Nfgen: Automatic non-linear function evaluation code generator for
  general-purpose mpc platforms.
\newblock In {\em Proceedings of the 2022 ACM SIGSAC Conference on Computer and
  Communications Security}, pages 995--1008, 2022.

\bibitem{goldreich1987play}
O~Goldreich, S~Micali, and A~Wigderson.
\newblock How to play any mental game.
\newblock In {\em Proceedings of the nineteenth annual ACM symposium on Theory
  of computing}, pages 218--229, 1987.

\bibitem{graham2021levit}
Benjamin Graham, Alaaeldin El-Nouby, Hugo Touvron, Pierre Stock, Armand Joulin,
  Herv{\'e} J{\'e}gou, and Matthijs Douze.
\newblock Levit: A vision transformer in convnet's clothing for faster
  inference.
\newblock In {\em Proceedings of the IEEE/CVF international conference on
  computer vision}, pages 12259--12269, 2021.

\bibitem{gupta2022llama}
Kanav Gupta, Deepak Kumaraswamy, Nishanth Chandran, and Divya Gupta.
\newblock Llama: A low latency math library for secure inference.
\newblock {\em Cryptology ePrint Archive}, 2022.

\bibitem{he2015delving}
Kaiming He, Xiangyu Zhang, Shaoqing Ren, and Jian Sun.
\newblock Delving deep into rectifiers: Surpassing human-level performance on
  imagenet classification.
\newblock In {\em Proceedings of the IEEE international conference on computer
  vision}, pages 1026--1034, 2015.

\bibitem{huang2022cheetah}
Zhicong Huang, Wen-jie Lu, Cheng Hong, and Jiansheng Ding.
\newblock Cheetah: Lean and fast secure $\{$two-party$\}$ deep neural network
  inference.
\newblock In {\em 31st USENIX Security Symposium (USENIX Security 22)}, pages
  809--826, 2022.

\bibitem{keller2020mp}
Marcel Keller.
\newblock Mp-spdz: A versatile framework for multi-party computation.
\newblock In {\em Proceedings of the 2020 ACM SIGSAC conference on computer and
  communications security}, pages 1575--1590, 2020.

\bibitem{keller2016mascot}
Marcel Keller, Emmanuela Orsini, and Peter Scholl.
\newblock Mascot: faster malicious arithmetic secure computation with oblivious
  transfer.
\newblock In {\em Proceedings of the 2016 ACM SIGSAC Conference on Computer and
  Communications Security}, pages 830--842, 2016.

\bibitem{keller2018overdrive}
Marcel Keller, Valerio Pastro, and Dragos Rotaru.
\newblock Overdrive: making spdz great again.
\newblock In {\em Advances in Cryptology--EUROCRYPT 2018: 37th Annual
  International Conference on the Theory and Applications of Cryptographic
  Techniques, Tel Aviv, Israel, April 29-May 3, 2018 Proceedings, Part III},
  pages 158--189. Springer, 2018.

\bibitem{keller2020effectiveness}
Marcel Keller and Ke~Sun.
\newblock Effectiveness of mpc-friendly softmax replacement.
\newblock {\em arXiv preprint arXiv:2011.11202}, 2020.

\bibitem{keller2022secure}
Marcel Keller and Ke~Sun.
\newblock Secure quantized training for deep learning.
\newblock In {\em International Conference on Machine Learning}, pages
  10912--10938. PMLR, 2022.

\bibitem{knott2021crypten}
Brian Knott, Shobha Venkataraman, Awni Hannun, Shubho Sengupta, Mark Ibrahim,
  and Laurens van~der Maaten.
\newblock Crypten: Secure multi-party computation meets machine learning.
\newblock {\em Advances in Neural Information Processing Systems},
  34:4961--4973, 2021.

\bibitem{krizhevsky2009learning}
Alex Krizhevsky, Geoffrey Hinton, et~al.
\newblock Learning multiple layers of features from tiny images.
\newblock 2009.

\bibitem{krizhevsky2017imagenet}
Alex Krizhevsky, Ilya Sutskever, and Geoffrey~E Hinton.
\newblock Imagenet classification with deep convolutional neural networks.
\newblock {\em Communications of the ACM}, 60(6):84--90, 2017.

\bibitem{lecun1998mnist}
Yann LeCun.
\newblock The mnist database of handwritten digits.
\newblock {\em http://yann. lecun. com/exdb/mnist/}, 1998.

\bibitem{lecun1998gradient}
Yann LeCun, L{\'e}on Bottou, Yoshua Bengio, and Patrick Haffner.
\newblock Gradient-based learning applied to document recognition.
\newblock {\em Proceedings of the IEEE}, 86(11):2278--2324, 1998.

\bibitem{li2022mpcformer}
Dacheng Li, Hongyi Wang, Rulin Shao, Han Guo, Eric Xing, and Hao Zhang.
\newblock Mpcformer: Fast, performant and private transformer inference with
  mpc.
\newblock In {\em The Eleventh International Conference on Learning
  Representations}, 2022.

\bibitem{li2019privpy}
Yi~Li and Wei Xu.
\newblock Privpy: General and scalable privacy-preserving data mining.
\newblock In {\em Proceedings of the 25th ACM SIGKDD International Conference
  on Knowledge Discovery \& Data Mining}, pages 1299--1307, 2019.

\bibitem{mcmahan2017communication}
Brendan McMahan, Eider Moore, Daniel Ramage, Seth Hampson, and Blaise~Aguera
  y~Arcas.
\newblock Communication-efficient learning of deep networks from decentralized
  data.
\newblock In {\em Artificial intelligence and statistics}, pages 1273--1282.
  PMLR, 2017.

\bibitem{mengmulti}
Yibai Meng, Shuxian Wang, and Alex Schedel.
\newblock Multi-gpu for piranha, a multiparty computation framework.

\bibitem{mohassel2018aby3}
Payman Mohassel and Peter Rindal.
\newblock Aby3: A mixed protocol framework for machine learning.
\newblock In {\em Proceedings of the 2018 ACM SIGSAC conference on computer and
  communications security}, pages 35--52, 2018.

\bibitem{mohassel2017secureml}
Payman Mohassel and Yupeng Zhang.
\newblock Secureml: A system for scalable privacy-preserving machine learning.
\newblock In {\em 2017 IEEE symposium on security and privacy (SP)}, pages
  19--38. IEEE, 2017.

\bibitem{naor2001efficient}
Moni Naor and Benny Pinkas.
\newblock Efficient oblivious transfer protocols.
\newblock In {\em SODA}, volume~1, pages 448--457, 2001.

\bibitem{nasr2019comprehensive}
Milad Nasr, Reza Shokri, and Amir Houmansadr.
\newblock Comprehensive privacy analysis of deep learning: Passive and active
  white-box inference attacks against centralized and federated learning.
\newblock In {\em 2019 IEEE symposium on security and privacy (SP)}, pages
  739--753. IEEE, 2019.

\bibitem{cuda11}
{NVIDIA}.
\newblock
  https://developer.nvidia.com/blog/cuda-11-6-toolkit-new-release-revealed-2/,
  2022.

\bibitem{openai2023gpt4}
OpenAI.
\newblock {GPT-4} technical report.
\newblock {\em arXiv preprint arXiv:2303.08774}, 2023.

\bibitem{radford2019language}
Alec Radford, Jeffrey Wu, Rewon Child, David Luan, Dario Amodei, Ilya
  Sutskever, et~al.
\newblock Language models are unsupervised multitask learners.
\newblock {\em OpenAI blog}, 1(8):9, 2019.

\bibitem{rathee2022secfloat}
Deevashwer Rathee, Anwesh Bhattacharya, Rahul Sharma, Divya Gupta, Nishanth
  Chandran, and Aseem Rastogi.
\newblock Secfloat: Accurate floating-point meets secure 2-party computation.
\newblock In {\em 2022 IEEE Symposium on Security and Privacy (SP)}, pages
  576--595. IEEE, 2022.

\bibitem{rathee2021sirnn}
Deevashwer Rathee, Mayank Rathee, Rahul Kranti~Kiran Goli, Divya Gupta, Rahul
  Sharma, Nishanth Chandran, and Aseem Rastogi.
\newblock Sirnn: A math library for secure rnn inference.
\newblock In {\em 2021 IEEE Symposium on Security and Privacy (SP)}, pages
  1003--1020. IEEE, 2021.

\bibitem{rathee2020cryptflow2}
Deevashwer Rathee, Mayank Rathee, Nishant Kumar, Nishanth Chandran, Divya
  Gupta, Aseem Rastogi, and Rahul Sharma.
\newblock Cryptflow2: Practical 2-party secure inference.
\newblock In {\em Proceedings of the 2020 ACM SIGSAC Conference on Computer and
  Communications Security}, pages 325--342, 2020.

\bibitem{rdma-core}
{RDMA-core contributors}.
\newblock {RDMA}-core.
\newblock \url{https://github.com/linux-rdma/rdma-core}, 2023.

\bibitem{remez1934determination}
Eugene~Y Remez.
\newblock Sur la d{\'e}termination des polyn{\^o}mes d’approxima\-tion de
  degr{\'e} donn{\'e}e.
\newblock {\em Comm. Soc. Math. Kharkov}, 10(196):41--63, 1934.

\bibitem{rombach2022high}
Robin Rombach, Andreas Blattmann, Dominik Lorenz, Patrick Esser, and Bj{\"o}rn
  Ommer.
\newblock High-resolution image synthesis with latent diffusion models.
\newblock In {\em Proceedings of the IEEE/CVF Conference on Computer Vision and
  Pattern Recognition}, pages 10684--10695, 2022.

\bibitem{rotaru2019marbled}
Dragos Rotaru and Tim Wood.
\newblock Marbled circuits: Mixing arithmetic and boolean circuits with active
  security.
\newblock In {\em Progress in Cryptology--INDOCRYPT 2019: 20th International
  Conference on Cryptology in India, Hyderabad, India, December 15--18, 2019,
  Proceedings}, pages 227--249. Springer, 2019.

\bibitem{gdrcopy}
Rong Shi, Sreeram Potluri, Khaled Hamidouche, Jonathan Perkins, Mingzhe Li,
  Davide Rossetti, and Dhabaleswar K~DK Panda.
\newblock Designing efficient small message transfer mechanism for inter-node
  mpi communication on infiniband gpu clusters.
\newblock In {\em 2014 21st International Conference on High Performance
  Computing (HiPC)}, pages 1--10. IEEE, 2014.

\bibitem{simonyan2014very}
Karen Simonyan and Andrew Zisserman.
\newblock Very deep convolutional networks for large-scale image recognition.
\newblock {\em arXiv preprint arXiv:1409.1556}, 2014.

\bibitem{tan2021cryptgpu}
Sijun Tan, Brian Knott, Yuan Tian, and David~J Wu.
\newblock Cryptgpu: Fast privacy-preserving machine learning on the gpu.
\newblock In {\em 2021 IEEE Symposium on Security and Privacy (SP)}, pages
  1021--1038. IEEE, 2021.

\bibitem{taylor1717methodus}
Brook Taylor.
\newblock {\em Methodus incrementorum directa \& inversa}.
\newblock Inny, 1717.

\bibitem{tolpegin2020data}
Vale Tolpegin, Stacey Truex, Mehmet~Emre Gursoy, and Ling Liu.
\newblock Data poisoning attacks against federated learning systems.
\newblock In {\em Computer Security--ESORICS 2020: 25th European Symposium on
  Research in Computer Security, ESORICS 2020, Guildford, UK, September 14--18,
  2020, Proceedings, Part I 25}, pages 480--501. Springer, 2020.

\bibitem{vaswani2017attention}
Ashish Vaswani, Noam Shazeer, Niki Parmar, Jakob Uszkoreit, Llion Jones,
  Aidan~N Gomez, {\L}ukasz Kaiser, and Illia Polosukhin.
\newblock Attention is all you need.
\newblock {\em Advances in neural information processing systems}, 30, 2017.

\bibitem{wagh2021falcon}
Sameer Wagh, Shruti Tople, Fabrice Benhamouda, Eyal Kushilevitz, Prateek
  Mittal, and Tal Rabin.
\newblock Falcon: Honest-majority maliciously secure framework for private deep
  learning.
\newblock {\em Proceedings on Privacy Enhancing Technologies}, 1:188--208,
  2021.

\bibitem{wang2023neural}
Chengyi Wang, Sanyuan Chen, Yu~Wu, Ziqiang Zhang, Long Zhou, Shujie Liu, Zhuo
  Chen, Yanqing Liu, Huaming Wang, Jinyu Li, et~al.
\newblock Neural codec language models are zero-shot text to speech
  synthesizers.
\newblock {\em arXiv preprint arXiv:2301.02111}, 2023.

\bibitem{watson2022piranha}
Jean-Luc Watson, Sameer Wagh, and Raluca~Ada Popa.
\newblock Piranha: A {GPU} platform for secure computation.
\newblock In {\em 31st USENIX Security Symposium (USENIX Security 22)}, pages
  827--844, 2022.

\bibitem{rdma_wiki}
{Wikipedia contributors}.
\newblock Remote direct memory access --- {Wikipedia}{,} the free encyclopedia,
  2022.
\newblock [Online; accessed 23-March-2023].

\bibitem{enwiki:newton}
{Wikipedia contributors}.
\newblock Newton's method --- {Wikipedia}{,} the free encyclopedia, 2023.

\bibitem{yao1982protocols}
Andrew~C Yao.
\newblock Protocols for secure computations.
\newblock In {\em 23rd annual symposium on foundations of computer science
  (sfcs 1982)}, pages 160--164. IEEE, 1982.

\end{thebibliography}
	\clearpage
	\appendix
	\section{Neural Network Structure}
\label{networks}
For the configuration of a convolution layer, we use a format of (filter channel, filter size, stride, padding). For the configuration of an average pooling layer, the format is (window size, stride).

\begin{table}[htbp]
\footnotesize
\centering
\begin{tabular}{ |c|c|c| } 
    \hline
    Layer $\#$ & Layer type & Configuration \\ 
    \hline
    1 & Fully-connected & $784\times 128$ \\
    \hline
    2 & ReLU & NA \\    
    \hline
    3 & Fully-connected & $128\times 128$ \\
    \hline
    4 & ReLU & NA \\
    \hline
    5 & Fully-connected & $128\times 10$ \\
    \hline
\end{tabular}
\vspace*{3mm}
\caption{Simple network with 722,848 parameters.}
\label{Network A}
\end{table}

\begin{table}[htbp]
\footnotesize
\centering
\begin{tabular}{ |c|c|c| } 
    \hline
    Layer $\#$ & Layer type & Configuration \\ 
    \hline
    1 & Convolution & $(20,5\times 5,1,0)$ \\
    \hline
    2 & Average pooling & $(2\times 2,2)$ \\
    \hline
    3 & ReLU & NA \\    
    \hline
    4 & Convolution & $(50,5\times 5,1,0)$ \\
    \hline
    5 & Average pooling & $(2\times 2,2)$ \\
    \hline
    6 & ReLU & NA \\
    \hline
    7 & Fully-connected & $800\times 500$ \\
    \hline
    8 & ReLU & NA \\
    \hline
    9 & Fully-connected & $500\times 10$ \\
    \hline
\end{tabular}
\vspace*{3mm}
\caption{LeNet with 430,500 parameters.}
\label{Network B}
\end{table}

\begin{table}[htbp]
\footnotesize
\centering
\begin{tabular}{ |c|c|c| } 
    \hline
    Layer $\#$ & Layer type & Configuration \\ 
    \hline
    1 & Convolution & $(96,11\times 11,4,9)$ \\
    \hline
    2 & Average pooling & $(3\times 3,2)$ \\
    \hline
    3 & ReLU & NA \\    
    \hline
    4 & Convolution & $(256,5\times 5,1,1)$ \\
    \hline
    5 & Average pooling & $(2\times 2,1)$ \\
    \hline
    6 & ReLU & NA \\
    \hline
    7 & Convolution & $(384,3\times 3,1,1)$ \\
    \hline
    8 & ReLU & NA \\
    \hline
    9 & Convolution & $(256,3\times 3,1,1)$ \\
    \hline
    10 & ReLU & NA \\
    \hline
    11 & Fully-connected & $256\times 256$ \\
    \hline
    12 & ReLU & NA \\
    \hline
    13 & Fully-connected & $256\times 256$ \\
    \hline
    14 & ReLU & NA \\
    \hline
    15 & Fully-connected & $256\times 10$ \\
    \hline
\end{tabular}
\vspace*{3mm}
\caption{AlexNet with 2,552,874 parameters.}
\label{Network C}
\end{table}

\begin{table}[htbp]
\footnotesize
\centering
\begin{tabular}{ |c|c|c| } 
    \hline
    Layer $\#$ & Layer type & Configuration \\ 
    \hline
    1 & Convolution & $(64,3\times 3,1,1)$ \\
    \hline
    2 & ReLU & NA \\    
    \hline
    3 & Convolution & $(64,3\times 3,1,1)$ \\
    \hline
    4 & Average pooling & $(2\times 2,2)$ \\
    \hline
    5 & ReLU & NA \\
    \hline
    6 & Convolution & $(128,3\times 3,1,1)$ \\
    \hline
    7 & ReLU & NA \\    
    \hline
    8 & Convolution & $(128,3\times 3,1,1)$ \\
    \hline
    9 & Average pooling & $(2\times 2,2)$ \\
    \hline
    10 & ReLU & NA \\
    \hline
    11 & Convolution & $(256,3\times 3,1,1)$ \\
    \hline
    12 & ReLU & NA \\    
    \hline
    13 & Convolution & $(256,3\times 3,1,1)$ \\
    \hline
    14 & Average pooling & $(2\times 2,2)$ \\
    \hline
    15 & ReLU & NA \\
    \hline
    16 & Convolution & $(512,3\times 3,1,1)$ \\
    \hline
    17 & ReLU & NA \\    
    \hline
    18 & Convolution & $(512,3\times 3,1,1)$ \\
    \hline
    19 & Average pooling & $(2\times 2,2)$ \\
    \hline
    20 & ReLU & NA \\
    \hline
    21 & Convolution & $(512,3\times 3,1,1)$ \\
    \hline
    22 & Average pooling & $(2\times 2,2)$ \\
    \hline
    23 & ReLU & NA \\
    \hline
    24 & Fully-connected & $512\times 256$ \\
    \hline
    25 & ReLU & NA \\
    \hline
    26 & Fully-connected & $256\times 256$ \\
    \hline
    27 & ReLU & NA \\
    \hline
    28 & Fully-connected & $256\times 10$ \\
    \hline
\end{tabular}
\vspace*{3mm}
\caption{Modified VGG16 (VGG16m) with 7,439,306 parameters.}
\label{Network D}
\end{table}

\section{Polynomial Coefficients}
\label{appendix:coefficients}
We give the coefficients of $f_1(z), f_2(z)$ in Table~\ref{coefficients},
where each polynomial is denoted as $k_4 \cdot z^4 + k_3 \cdot z^3 + k_2 \cdot z^2 +k_1 \cdot z+ k_0$.

\begin{table}[htbp]
    \footnotesize
    \centering
    \renewcommand{\arraystretch}{1.05}
    \begin{tabular}{ |c|c|c| } 
     \hline
     Coefficient & $f_1(z)$ & $f_2(z)$ \\ 
     \hline
     $k_4$  & 0.01353417  & -0.268344 \\ 
     $k_3$  & 0.05201146  & 0.496404 \\ 
     $k_2$  & 0.24144276  & -0.726980 \\ 
     $k_1$  & 0.69300383  & 1.442547 \\ 
     $k_0$  & 1.00000259  & 0 \\
     $k_4'$ & 0.015625    & -0.25 \\
     $t_3$  & 0.96074360  & -0.46246981 \\
     $t_2$  & 6.15066406  & 0.712932283 \\
     $t_1$  & 24.02000383 & -0.366125013 \\
     $t_0$  & 23.85013773 & -0.858977912 \\
     \hline
    \end{tabular}
 \vspace*{3mm}
 \caption{Coefficients of Polynomial Approximation}
 \label{coefficients}
\end{table}

\section{Security Analysis of Algorithms}
\label{appendix:security}

We assume a semi-honest adversary $\mathcal{A}$ who can corrupt any subset of the parties (no more than $n-1$, where $n$ is the number of parties) before the protocol begins.
We say that a protocol $\pi$ securely implements a functionality $F$ if any adversary $\mathcal{A}$ with probabilistic polynomial time cannot distinguish
between a protocol $\pi$ and a functionality $F$ attached to a simulator $\mathcal{S}$.

\para{Security definition.}
Security is defined by comparing a \emph{real} and \emph{ideal} interaction.
We denote $\text{Real}_{\pi,\mathcal{Z}(\mathcal{A},\lambda)}$
the bit output by environment $\mathcal{Z}$ after the execution of protocol $\pi$ on security parameter $\lambda$, which is a real interaction involving protocol $\pi$. We denote $\text{Ideal}_{F,\mathcal{Z}(S)}$ the bit output by $\mathcal{Z}$ after the execution of trusted
functionality.
As an interface, $\mathcal{S}$, also known as a simulator, can access $F$ and simulate a run of $\pi$ with the same input-output behavior.
$\mathcal{Z}$ can observe the inputs of $\mathcal{A}$.

We say that $\pi$ realizes $F$ if there exists a simulator $\mathcal{Z}$ s.t. for all environment $\mathcal{Z}$ holds that
\begin{equation*}
    |\Pr[\text{Real}_{\pi,\mathcal{Z}(\mathcal{A},\lambda)}] - \Pr[\text{Ideal}_{F,\mathcal{Z}(S)}]|\leq 2^{-\lambda}
\end{equation*}
The $F_{\text{ABB}}$-Hybrid model is the real-world model where the parties have access to an ideal functionality $F_{\text{ABB}}$.
We say that $\pi$ realizes $F$ in the $F_{\text{ABB}}$-Hybrid model if there exists a simulator $\mathcal{Z}$ s.t. for all environment $\mathcal{Z}$ holds that
\begin{equation*}
    |\Pr[\text{Hybrid}^{F_{\text{ABB}}}_{\pi,\mathcal{Z}(\mathcal{A},\lambda)}] - \Pr[\text{Ideal}_{F,\mathcal{Z}(S)}]|\leq 2^{-\lambda},
\end{equation*}
where $\text{Hybrid}^{F_{\text{ABB}}}_{\pi,\mathcal{Z}(\mathcal{A},\lambda)}$ is the bit output by environment $\mathcal{Z}$ after the execution of protocol $\pi$ based on $F_{\text{ABB}}$-Hybrid model on security parameter $\lambda$.
In this paper, we only give the proof for security of the Reciprocal algorithm based on the $F_{\text{ABB}}$-Hybrid model, and the security proofs for the other algorithms are similar. 

\begin{mybox}
\centerline{Functionality $F_{\text{Reciprocal}}$}
$\textbf{INPUT: }\share{x,p}{A},x\neq 0, \text{bit length } l,\text{precision } p\\
\textbf{OUTPUT: }\share{y,p}{A},y\approx\frac{1}{x}$
\end{mybox}

\begin{theorem}
    In the $F_{\text{ABB}}$-Hybrid model, our reciprocal algorithm \ref{alg:ssnr} realizes the ideal functionality $F_{\text{Reciprocal}}$
    in presence of a semi-honest adversary.
\end{theorem}

\begin{proof}[Proof]
Consider the worst-case scenario: except for an honest party $P_h$, all the other $n-1$ parties are corrupted.

\para{Initialize.} Now we construct a simulator $S$ for adversary $\mathcal{A}$ and any environment $\mathcal{Z}$
who chooses all parties' inputs and sees all intermediate results and outputs obtained by an adversary.
In an ideal world, $S$ simulates the honest $P_h$'s input share with uniformly random values $r$.
Note that $S$ does not need to alter the input of the corrupted adversary. Now the actual input changes from $[x]$ to $[x']$.
$x'$ is from a uniform distribution.

\para{Indistinguishability.} We now argue the indistinguishability of the real and ideal executions to an environment $\mathcal{Z}$.
Recall that $\mathcal{Z}$ can choose the inputs for all parties. The view of $\mathcal{Z}$ then consists of these inputs ($\share{x,p}{A}$, $\share{x',p}{A}$),
the messages received by the adversary (namely shares of intermediate results obtained by the adversary),
and output ($\share{y,p}{A}$) shares obtained by the adversary.

For line~\ref{rec:decompose} of Algorithm~\ref{alg:ssnr}, we use $\code{Decompose}$ primitive by $F_{\text{ABB}}$.
From the security of the arithmetic secret sharing, $\mathcal{Z}$ can only see uniformly random strings of bits. 
In an ideal world, the original input $[x']$ is uniformly random.
Thus, in $\mathcal{Z}$'s view, the $\share{x_{0:l-1}}{B}$ in the real world is indistinguishable from the $\share{x'_{0:l-1}}{B}$ of the ideal world.

For line~\ref{rec:xor_start} - \ref{rec:lmo}, in the real world, all messages received by the adversary are output by
series of $F_{\text{ABB}}$ primitives. So $\share{x_{0:l-1}}{B}$ and $\share{x'_{0:l-1}}{B}$ are indistinguishable,
which leads to the indistinguishability of $\share{x_{0:l-2}}{B}$ and $\share{x'_{0:l-2}}{B}$.
In both the real world and ideal world, in $\mathcal{Z}$'s view, these strings of bits are uniformly random.

For line~\ref{rec:swap_start} - \ref{rec:swap_end}, since the output of the previous step is indistinguishable
and the bit strings are uniformly random, the operation of swapping their order and setting them to zero is indistinguishable.

For line~\ref{rec:compose} - \ref{rec:flip_back}, we again use the primitives by $F_{\text{ABB}}$ and this process is indistinguishable.
\end{proof}

\section{Baseline Exponentiation Algorithm}
\label{appendix:baseline_exp}
Algorithm~\ref{alg:qte} shows the exponentiation algorithm in \cite{keller2022secure}.

\begin{algorithm}[tb]
\small
\caption{Exponentiation of \cite{keller2022secure}.}\label{alg:qte}
\KwData{$\share{x,p}{A}$,\text{ bit length }$l$,\text{ precision }$p$}
\KwPara{Taylor approximation $g(z)$}
\KwResult{$\share{y,p}{A}\approx \share{e^x,p}{A}$}
$\share{x',p}{A}\leftarrow\code{Truncate}(\share{x,p}{A}\cdot[\log_2 e]_p$\\
$\share{x_{0:l-1}}{B}\leftarrow\code{Decompose}(\share{x',p}{A})$\\
$\share{b}{B}\leftarrow\share{x_{0:l-1}}{B}+\share{l-p-1}{B}>0$\\
$c\leftarrow\ceil{\log_2(l-p)}$\label{exp_keller:compare}\\
$\share{x_{p,p+c-1}}{A}\leftarrow\code{Bit2A}(\share{x_{p,p+c-1}}{B}$\\
$\share{I}{A}\leftarrow\prod_{i=0}^{c-1} \Big( 2^{2^{i}}\share{x_{p+i}}{A}-\share{x_{p+i}}{A}+1\Big)$\\
$\share{t,p}{A}\leftarrow\code{Compose}(\share{x_{0:p-1}}{B}$\\
$\share{D,p}{A}\leftarrow\code{Evaluate}(\share{g(t),p}{A})$\\
$\share{y_{+},p}{A}\leftarrow\share{I}{A}\cdot\share{D,p}{A}$\\
$\share{y_{-},p}{A}\leftarrow\code{Truncate}(\share{y_{+},p}{A},2^c)$\label{exp_keller:correct}\\
$\share{s}{A}\leftarrow\code{Bit2A}(\share{x_{l-1}}{B})$\\
$\share{y,p}{A}\leftarrow\share{s}{A}(\share{y_{-},p}{A}-\share{y_{+},p}{A})+\share{y_{+},p}{A}$\\
$\share{y,p}{A}\leftarrow(1-\code{Bit2A}(\share{b}{B}))\cdot\share{y,p}{A}$\label{exp_keller:result}
\end{algorithm}

\section{Attention-Specific Exponentiation}
\label{appendix:exponentiation}

Algorithm~\ref{alg:attention_exp} details how we evaluate exponentiation in an attention block when we perform Transformer-based model inference.
Note that, line~\ref{sexp:bias} - \ref{sexp:decompose} work with bitlength $l/2$, while line~\ref{sexp:bit2a} and the following work with bitlength $l$.
That is, in the case $l = 64$, we convert the \code{Fixed32} input to a \code{Fixed64} output for free.
Line~\ref{sexp:sign} uses $\share{x_{l/2-1}}{B}$ as the sign bit, due to the fact that $\share{x_{l/2-1}}{B}$ equals $\share{x_{l-1}}{B}$ when $l/2$ bits are enough to represent the input.

\begin{algorithm}[tb]
\small
\caption{Exponentiation in attention.}\label{alg:attention_exp}
\KwData{\text{Bit length }$l$,\text{ precision }$p$, $\share{x',p}{A}$ where $x'$ equals $\log_2e \cdot x$ and is of bitlength $l/2$}
\KwPara{Remez approximation $f_1'(z)$ (the completing square version of $f_1(z)$)}
\KwResult{$\share{y,p}{A}\approx \share{e^x,p}{A}$ of bitlength $l$}
$\share{x'',p}{A}\leftarrow [p]_{p}+\share{x',p}{A}$ \label{sexp:bias}\\
$\share{x_{0:l/2-1}}{B}\leftarrow\code{Decompose}(\share{x'',p}{A})$\label{sexp:decompose}\\
$c\leftarrow\ceil{\log_2 p}$\\
$\share{x_{p:p+c-1}}{A}\leftarrow\code{Bit2A}(\share{x_{p:p+c-1}}{B})$ \label{sexp:bit2a}\\
$\share{I}{A}\leftarrow\prod_{i=0}^{c-1} \Big( 2^{2^{i}}\share{x_{p+i}}{A}-\share{x_{p+i}}{A}+1\Big)$\\
$\share{t,p}{A}\leftarrow\code{Compose}(\share{x_{0:p-1}}{B})$\\
$\share{y,p}{A}\leftarrow\share{I}{A}\cdot\code{Evaluate}(\share{f_1'(t),p}{A})$\\
$\share{s}{A}\leftarrow 1-\code{Bit2A}(\share{x_{l/2-1}}{B})$ \label{sexp:sign}\\
$\share{y,p}{A}\leftarrow\code{Truncate}(\share{s}{A}\cdot\share{y,p}{A},p)$
\end{algorithm}

\end{document}